\documentclass[11pt]{article}

\usepackage{amsmath,enumerate,enumitem}
\usepackage{amssymb}
\usepackage{amsthm}
\usepackage{graphicx}
\usepackage{tikz}
\usepackage{xcolor}
\usepackage{hyperref}
\usepackage{mathtools}
\usepackage{url}
\mathtoolsset{showonlyrefs}
\usepackage{enumitem}

\allowdisplaybreaks

\usepackage[margin=3.5cm]{geometry}

\def\p{\partial}
\def\bs{\backslash}
\def\subset{\subseteq}

\def\Z{\mathbb{Z}}

\def\eps{\varepsilon}
\def\del{\delta}
\def\gam{\gamma}

\def\cC{\mathcal{C}}

\def\cE{\mathcal{E}}

\def\cG{\mathcal {G}}

\def\cI{\mathcal{I}}

\def\cL{\mathcal {L}}
\def\cO{\mathcal{O}}
\def\cP{\mathcal{P}}

\def\cX{\mathcal{X}}

\def\1{\mathbf{1}}

\def\lam {\lambda}

\def\gam{\gamma}
\def\Gam{\Gamma}
\def\eps{\varepsilon}
\def\del{\delta}

\def\colset{\mathcal X}

\newtheorem*{theorem*}{Theorem}
\newtheorem{theorem}{Theorem}
\newtheorem{lemma}[theorem]{Lemma}
\newtheorem{cor}[theorem]{Corollary}
\newtheorem{defn}[theorem]{Definition}
\newtheorem*{defn*}{Definition}

\newtheorem*{prop*}{Proposition}

\newtheorem*{conj*}{Conjecture}
\newtheorem{claim}[theorem]{Claim}
\newtheorem{question}[theorem]{Question}
\newtheorem*{fact*}{Fact}
\newtheorem{fact}[theorem]{Fact}
\newtheorem{remark}{Remark}

\begin{document}
\title{Algorithms for \#BIS-hard problems \\ on expander graphs}

\author{Matthew Jenssen\thanks{University of Birmingham, m.jenssen@bham.ac.uk.}\and Peter Keevash\thanks{University of Oxford, keevash@maths.ox.ac.uk.\newline Research supported in part by ERC Consolidator Grant 647678.} \and Will Perkins\thanks{University of Illinois at Chicago, math@willperkins.org. \newline Research supported in part by EPSRC grant EP/P009913/1 and NSF Career award DMS-1847451.}}% }}

\date{March 23, 2020}

\maketitle

\begin{abstract}
We give an FPTAS and an efficient sampling algorithm for the high-fugacity hard-core model on bounded-degree bipartite expander graphs and the low-temperature ferromagnetic Potts model on bounded-degree expander graphs.  The results apply, for example, to random (bipartite) $\Delta$-regular graphs, for which no efficient algorithms were known for these problems (with the exception of the Ising model) in the non-uniqueness regime  of the infinite $\Delta$-regular tree.  We also find efficient counting and sampling algorithms for proper $q$-colorings of random $\Delta$-regular bipartite graphs when $q$ is sufficiently small as a function of $\Delta$. 
\end{abstract}

\section{Introduction}

There are two natural computational problems associated to a statistical physics spin model on a graph: the \textit{approximate counting} problem of approximating the partition function of the model and the \textit{sampling} problem of  obtaining a random spin configuration approximately distributed according to the model.

A prominent example is  the hard-core model of weighted independent sets. For a graph $G$ and fugacity parameter $\lam>0$, the hard-core model is the probability distribution $\mu_{G,\lam}$ on the collection $\cI(G)$ of independent sets of $G$ given by
\begin{align*}
\mu_{G,\lam}(I) &= \frac{\lam^{|I|}}{Z_G(\lam)} 
\end{align*}
where 
\begin{align*}
Z_G(\lam) &= \sum_{I \in \cI(G)} \lam^{|I|}
\end{align*}
is the hard-core partition function (also known as the independence polynomial in graph theory).

A \textit{fully polynomial-time approximation scheme} (FPTAS) is an algorithm that for every $\eps>0$ outputs an $\eps$-relative approximation to $Z_G(\lam)$ (that is, a number $\hat Z$ so that $e^{- \eps} \hat Z \le Z_G(\lam) \le e^{\eps} \hat Z$) and runs in time polynomial in $|V(G)|$ and $1/\eps$.  A \textit{polynomial-time sampling algorithm} is a randomized algorithm that for every $\eps>0$ runs in time polynomial in $|V(G)|$ and $1/\eps$ and outputs an independent set $I $ with distribution $\mu_{\text{alg}}$ so that $\| \mu_{G,\lam} - \mu_{\text{alg}} \|_{TV} <\eps$ (this terminology is not completely standard, but it defines the natural sampling equivalent of the running-time guarantees of an FPTAS). 

The computational complexity of the approximate counting and sampling problems for the hard-core model is well understood for bounded-degree graphs.  For graphs of maximum degree at most $\Delta$, when $\lam<\lam_c (\Delta) = \frac{(\Delta-1)^{\Delta-1}   }{ (\Delta-2)^\Delta  }$, there is an FPTAS and a polynomial-time sampling algorithm due to Weitz~\cite{weitz2006counting} (see also the recent~\cite{anari2020spectral}); whereas when $\lam > \lam_c(\Delta)$ both computational problems are hard: there is no polynomial-time algorithm unless NP$=$RP~\cite{sly2010computational,sly2014counting,galanis2016inapproximability}.  The value $\lam_c(\Delta)$ is the \textit{uniqueness threshold} of the hard-core model on the infinite $\Delta$-regular tree~\cite{kelly1985stochastic}. 

On the other hand, if we restrict ourselves to bipartite graphs, then classifying the computational complexity of these tasks are open problems.  The class \#BIS is the class of problems polynomial-time equivalent to approximating the number of independent sets of a bipartite graph~\cite{dyer2004relative}, and many interesting approximate counting and sampling problems have been shown to be \#BIS-hard~\cite{goldberg2012approximating,chebolu2012complexity,galanis2016ferromagnetic} (that is, at least as hard as approximating the number of independent sets in a bipartite graph).  In particular, Cai, Galanis, Goldberg, Guo, Jerrum, {\v{S}}tefankovi{\v{c}}, and Vigoda~\cite{cai2016hardness} showed that  for all $\Delta\ge 3$ and all $\lam > \lam_c(\Delta)$, it is \#BIS-hard to approximate the hard-core partition function at fugacity $\lam$ on a bipartite graph of maximum degree $\Delta$.  Resolving the complexity of \#BIS is a major open problem in the field of approximate counting.

One direction for partial progress on any intermediate complexity class is to find subclasses of instances for which the problem is tractable (e.g. results showing that the Unique Games problem is tractable on expander graphs~\cite{arora2008unique,makarychev2010play}). For \#BIS, we would like to find subclasses of bipartite graphs on which we can efficiently approximate the number of independent sets or the hard-core partition function.  One example is the algorithm of Liu and Lu~\cite{liu2015fptas}  which works when $\lam < \lam_c(\Delta)$ and one side of the bipartition has maximum degree $\Delta $ but the other side of the bipartition is allowed unbounded degree.

Recently, Helmuth, Perkins, and Regts~\cite{helmuth2018contours} gave efficient algorithms for the hard-core model at high fugacity on the torus $(\Z/n\Z)^{d}$ and subsets of the lattice $\Z^d$ with certain boundary conditions. The algorithms are based on contour models from Pirogov-Sinai theory~\cite{pirogov1975phase} along with the cluster expansion and Barvinok's polynomial interpolation method~\cite{barvinok2016computingP,barvinok2017combinatorics}.  

Here we extend this approach and use abstract polymer models and the cluster expansion to give efficient approximate counting and sampling algorithms for all bounded-degree, bipartite expander graphs at sufficiently high fugacity. We say that a bipartite graph $G=(\cO,\cE, E)$ is a \emph{bipartite $\alpha$-expander} if $|\partial S|\geq (1+\alpha) |S|$ for all $S\subseteq \cO$ with $|S|\le |\cO|/2$ and all $S\subseteq \cE$ with $|S|\le |\cE|/2$. Here $\p S$ denotes the \emph{vertex boundary} of $S$, the set of all vertices in $V(G)\backslash S$ with a neighbor in $S$.

\begin{theorem}
\label{thmhard-core}
There exists an absolute constant $C$ such that for every $\alpha >0$, $\Delta \ge 3$, and  any $\lam > C\Delta^{4/\alpha}$, there exists an FPTAS and a polynomial-time sampling algorithm for the hard-core model at fugacity $\lam$ on bipartite $\alpha$-expander graphs of maximum degree $\Delta$.
\end{theorem}

In general the running time of our approximate counting algorithm in Theorem~\ref{thmhard-core} (and the theorems that follow) is $(n/\eps)^{O(\log \Delta)}$ for an $\eps$-relative approximation to $Z_G(\lam)$ for an $n$-vertex graph $G$.  Since $\Delta$ is fixed this running time is polynomial in $n$ and $1/\eps$, but one might hope to improve the dependence on $\Delta$.  In fact, the larger $\lam$ is and the larger $\alpha$ is, the faster the running time, and for large enough values of the parameters the algorithm runs in almost linear time.  We discuss this more below.  See also~\cite{PolymerMarkov} in which the authors find faster sampling algorithms in similar settings using abstract polymer models and Markov chains.  

We can extend our methods to obtain efficient counting and sampling algorithms for the hard-core model on random regular bipartite graphs for much smaller values of $\lam$, all the way down to $\lam = \Omega\left( \frac{\log^2 \Delta}{\Delta}\right)$.

Let $\cG^{\text{bip}}(n , \Delta)$ be the set of all $\Delta$-regular bipartite graphs on $n$ vertices (where $n$ is even), and let $\mathbf G^{\text{bip}}_{n,\Delta}$ be a uniformly chosen graph from $\cG^{\text{bip}}(n , \Delta)$.   We say that a property holds for \textit{almost every} $\Delta$-regular bipartite graph if the property holds with  probability $\to 1$ as $n\to \infty$ for  $\mathbf G^{\text{bip}}_{n,\Delta}$.
\begin{theorem}
\label{thmHCrandombip}
There exists a  constant $\Delta_0$ so that for every $\Delta\ge \Delta_0$ and all $\lam > \frac{50 \log^2 \Delta}{\Delta}$, there is an FPTAS and an efficient sampling algorithm for the hard-core model at fugacity $\lam$ on almost every $\Delta$-regular bipartite graph.  The running time of the approximate counting algorithm is $(n/\eps)^{1+O(\log^2 \Delta/\Delta)}$. 
\end{theorem}
In particular, by setting $\lam=1$ this gives an FPTAS for counting the total number of independent sets in random $\Delta$-regular bipartite graphs for large enough $\Delta$.   

\begin{remark}\label{rmkAlg}
 The guarantees of the algorithm of Theorem~\ref{thmHCrandombip} (and of those below in Corollary~\ref{corPottsrandom} and Theorem~\ref{thmColor}) are of the following form: the counting algorithm will output an estimate of the partition function for every $\Delta$-regular bipartite graph. With probability $1-o(1)$ over the choice of graph $\mathbf G^{\text{bip}}_{n,\Delta}$ the estimate will satisfy the guarantees of an FPTAS.  In fact what is needed is that the graph satisfy some expansion conditions which hold with probability $1-o(1)$ for $\mathbf G^{\text{bip}}_{n,\Delta}$.  It would be desirable to find  algorithms with a stronger type of guarantee: an algorithm could give no answer on a vanishing fraction of graphs but if it gives an answer, the answer must satisfy the stated guarantees (e.g. the sampling algorithm in~\cite{blanca2018sampling}).  We find such an algorithm for the Potts model on random regular graphs below (Corollary~\ref{corPottsrandom}) by certifying expansion using spectral methods.  We believe such an algorithm can be found for the hard-core model as well (also using spectral results, e.g.~\cite{tanner1984explicit,brito2018spectral}), but certifying the specific expansion properties used in Lemma~\ref{lemsmallrandom} below seems more difficult, and so we do not pursue it here.   
\end{remark}

To the best of our knowledge, no efficient counting or sampling algorithms for random regular bipartite graphs were previously known for any $\lam > \lam_c(\Delta)$.

Moreover, for any $\del>0$, taking $\Delta= \Omega(\log^2(1/\del)/\del )$ gives an approximate counting algorithm running in time $O((n/\eps)^{1+\del})$.

\subsection{The Potts Model}
Given a graph $G$ and $q\in\mathbb{N}$, let $\Omega = [q] = \{1,\dots,q\}$ and let $\Omega^{V(G)}$ be the set of all colorings $\omega: V(G)\to \Omega$.
Given $\omega\in\Omega^{V(G)}$ let $m(G,\omega)$ denote the number of monochromatic edges induced by the coloring $\omega$, that is
\begin{align*}
m(G,\omega):=\sum_{\{i,j\}\in E(G)}\del_{\omega(i),\omega(j)}
\end{align*}
where $\del$ is the Kronecker delta function.
The $q$-color Potts model on $G$ at inverse temperature $\beta$ is the probability distribution on $\Omega^{V(G)}$ defined by
\begin{align*}
\mu_{G,q,\beta}(\omega) &= \frac{e^{ \beta\cdot m(G,\omega)}} {Z_{G,q}(\beta) },  \quad \omega \in \Omega^{V(G)}
\end{align*}
where
\[
Z_{G,q}(\beta):=\sum_{\omega\in\Omega^{V(G)}}e^{\beta\cdot m(G,\omega)}
\]
is the Potts model partition function.  When $\beta>0$ the model is \textit{ferromagnetic} (monochromatic edges preferred) and when $\beta<0$ the model is \textit{antiferromagnetic} (bichromatic edges preferred). 

Goldberg and Jerrum \cite{goldberg2012approximating} showed that approximating the partition function of the ferromagnetic
Potts model for $q\ge3$ is \#BIS-hard. Later,
Galanis, {\v{S}}tefankovi{\v{c}}, Vigoda, and Yang~\cite{galanis2016ferromagnetic} refined this result by showing that this task is \#BIS-hard on graphs of maximum degree $\Delta$ when $\beta > \beta_o(q,\Delta)$, the order/disorder threshold of the infinite $\Delta$-regular tree (see~\cite{galanis2016ferromagnetic} for a precise definition of $\beta_o(q,\Delta)$; in particular, $\beta_o(q,\Delta) > \beta_c(q,\Delta)$, the uniqueness threshold on the infinite $\Delta$-regular tree). 

Our next theorem gives efficient counting and sampling algorithms for the ferromagnetic Potts model at low enough temperatures on expander graphs.    We say that a graph $G$ is an \emph{$\alpha$-expander} if $|\p_e(S)|\geq \alpha |S|$ for all subsets $S\subseteq V(G)$ with $|S|\le |V(G)|/2$.
Here $\p_e(S)$ denotes the \emph{edge boundary} of $S$, the set of edges of $G$ with one endpoint in $S$ and the other in $V(G)\backslash S$.

\begin{theorem}\label{thm:pottspartition}
For all $\alpha>0$, $\Delta \ge 3$, $q\geq2$ and $\beta \ge   \frac{ 4+   2\log (q \Delta)   }{ \alpha  }$, there is an FPTAS and polynomial-time sampling algorithm for the $q$-color ferromagnetic Potts model at inverse temperature $\beta$ on all $\alpha$-expander graphs of maximum degree $\Delta$.
\end{theorem}

Our algorithms apply to  the Potts model on the random $\Delta$-regular graph as well. Let $\cG(n, \Delta)$ be the set of all $\Delta$-regular graphs on $n$ vertices, and let  $\mathbf G_{n,\Delta}$ be a uniformly chosen graph from $\cG(n, \Delta)$ (as long as this set is non-empty).  We say that a property holds for \textit{almost every} $\Delta$-regular graph if the property holds with  probability $\to 1$ as $n\to \infty$ for  $\mathbf G_{n,\Delta}$.

\begin{cor}
\label{corPottsrandom} 
There is an absolute constant $C>0$ so that for every $\Delta \ge 3$, $q\ge 2$,  and all $\beta > \frac{C\log( q\Delta)}{\Delta}$, there is an FPTAS and polynomial-time sampling algorithm for the ferromagnetic Potts model at inverse temperature $\beta$ on almost every $\Delta$-regular graph.  Moreover, there is a polynomial-time algorithm to certify conditions on a $\Delta$-regular graph $G$ that suffice for the guarantees of the FPTAS and sampling algorithm, and these conditions hold with probability $1-o(1)$ for $\mathbf G_{n,\Delta}$.
\end{cor}
Again to the best of our knowledge no efficient counting or sampling algorithms were known previously for the $q \ge 3$ Potts model on random regular graphs for $\beta $ above the uniqueness threshold $\beta_c(q,\Delta)$ of the infinite $\Delta$-regular tree.  The lower bound on $\beta$ needed in Corollary~\ref{corPottsrandom} is necessarily above the ordering threshold $\beta_o(q,\Delta)$ on the infinite tree since our proof involves a phase coexistence result (Lemma~\ref{lem:starapprox}). We note however that our lower bound on $\beta$ is within a constant factor of $\beta_o$ for $\Delta$ fixed and $q$ large (since $\beta_o(q,\Delta) = (1+o_q(1)) \frac{2 \log q}{\Delta}$ as $q \to \infty$~\cite{galanis2016ferromagnetic}).

\subsection{Counting and sampling proper colorings on bipartite graphs}
\label{secIntroColorings}

   For $q \ge 3$, let $\colset_q(G)$ be the set of all proper $q$-colorings of $G$.  Let $Z_G(q) = | \colset_q(G)|$ and let $\mu_{G,q}$ be the uniform distribution on $\colset_q(G)$. In particular, if $G$ is bipartite then $\colset_q(G)$ is guaranteed to be non-empty.  We would like to approximate $Z_G(q)$ and sample from $\mu_{G,q}$.  Galanis, Stefankovic, Vigoda, and Yang~\cite{galanis2016ferromagnetic} show that these problems are \#BIS-hard on bipartite graphs of maximum degree $\Delta$ when $q \le \Delta/(2 \log \Delta)$. Here we show that under the stricter condition that $q \le c \sqrt{\Delta}/\log^{2} \Delta$ (or equivalently $\Delta \ge C q^2 \log^2 q$), these problems are tractable on random regular bipartite graphs. This solves Conjecture~1 from the earlier extended abstract of this paper~\cite{JenssenAlgorithmsSODA}.

\begin{theorem}
\label{thmColor}
There is an absolute constant $C>0$ so that for all $q \ge 3$, and for  $\Delta \ge C q^2 \log^2 q$,  there is an FPTAS and polynomial-time sampling algorithm for proper colorings on almost every $\Delta$-regular bipartite graph.
\end{theorem}

Adapting techniques from~\cite{helmuth2018contours,JenssenAlgorithmsSODA}, Liao, Lin, Lu, and Mao~\cite{liao2019counting} have recently and independently proved similar results to Theorems~\ref{thmHCrandombip} and~\ref{thmColor} with slightly stronger conditions needed on $\lam$ and $\Delta$ respectively to obtain efficient algorithms.  

\subsection{Discussion}

We take $\Delta \ge 3$ in all of our theorems since computing the relevant partition functions exactly on paths and cycles takes linear time.  

We also note that for the $q=2$ case of the Potts mode (the Ising model), efficient algorithms are known for \textit{all} graphs and all temperatures: the approximate counting algorithm of Jerrum and Sinclair~\cite{jerrum1993polynomial}, turned into a sampling algorithm via self-reducibility by Randall and Wilson~\cite{randall1999sampling}; see also the recent proof of Guo and Jerrum~\cite{guo2018} showing polynomial-time mixing of the random-cluster dynamics.

The approximate counting and sampling problems for the hard-core and Potts models on random graphs have received considerable attention, with positive algorithmic results in the low-fugacity, high-temperature uniqueness regimes of the infinite $\Delta$-regular tree, and some negative algorithmic results, in the form of torpid mixing of certain Markov chains, in the high-fugacity and low-temperature regimes.  Our results are novel in providing positive algorithmic results in the high-fugacity and low-temperature regimes.

In the low fugacity regime with $\lam< \lam_c(\Delta)$, Weitz's algorithm applies  to  $\mathbf G^{\text{bip}}_{n,\Delta}$. Efthymiou, Hayes, {\v{S}}tefankovic, Vigoda, and Yin~\cite{efthymiou2016convergence} have also shown that the Glauber dynamics have mixing time $O( n \log n)$ on $\mathbf G^{\text{bip}}_{n,\Delta}$ for $\lam < (1-\eps(\Delta)) \lam_c(\Delta)$ for some $\eps (\Delta) \to 0$ as $\Delta \to \infty$.   For $\lam > \lam_c(\Delta)$, the Glauber dynamics for the hard-core model on  $\mathbf G^{\text{bip}}_{n,\Delta}$ are known to mix slowly~\cite{mossel2009hardness}, and perhaps Theorem~\ref{thmHCrandombip} can be improved to work for all $\lam > \lam_c(\Delta)$, though this would likely require new ideas.

For the Potts model, a natural conjecture for the optimal bound on $\beta$ for the particular polymer-based algorithm we use here is the order/disorder transition point $\beta_o(q,\Delta) = \log \frac{q-2}{(q-1)^{1-2/\Delta}-1}$. Galanis, {\v{S}}tefankovi{\v{c}}, Vigoda, and Yang~\cite{galanis2016ferromagnetic} have shown that the Swendsen-Wang dynamics mix slowly at $\beta_o(q,\Delta)$ on the random $\Delta$-regular graph (for $q$ large enough).  In fact their analysis shows that the approximation lemmas we use below fail for $\mathbf G_{n,\Delta}$ and $\beta\le \beta_o(q,\Delta)$. The bound we obtain in Corollary~\ref{corPottsrandom} is at worst a factor of order $\log \Delta$ away from this natural barrier and matches up to a constant factor when $q$ and $\Delta$ are polynomially related.  

In the high-temperature regime, Blanca, Galanis, Goldberg, {\v{S}}tefankovic, Vigoda, and Yang~\cite{blanca2018sampling} have recently given an efficient algorithm to obtain an $n^{-c}$-approximate sample (that is, a sample within total variation distance $n^{-c}$ for $n$-vertex graphs for some constant $c>0$) from the Potts model (ferromagnetic and anti-ferromagnetic) on $\mathbf G_{n,\Delta}$ when the parameters lie the uniqueness regime for the infinite  $\Delta$-regular tree.

While efficient counting and sampling algorithms for these problems on random regular graphs were previously only known for the uniqueness regime, the probabilistic properties of these models are well understood at all fugacities and temperatures.  Sly and Sun~\cite{sly2014counting} showed that the limiting free energy (the normalized log partition function) of \textit{any} sequence of locally tree-like bipartite graphs converges to the replica symmetric solution predicted by the cavity method from statistical physics.  This result applies in particular to the hard-core model on random bipartite $\Delta$-regular graphs. Dembo, Montanari, Sly, and Sun~\cite{dembo2014replica} then showed that the limiting free energy of the ferromagnetic Potts model on a sequence of graphs converging locally to the infinite $\Delta$-regular tree is given by the replica symmetric solution from the cavity method.

For other recent work on algorithms for the Potts model and low temperature and the hard-core model at high fugacity, see~\cite{barvinok2017weighted, PolymerMarkov, cannon2019counting}.  For further algorithmic applications of abstract polymer models and the cluster expansion see~\cite{casel2019zeros}.

\subsection{Proof ideas}

Our  main technical contribution is to show that the hard-core, Potts and coloring models are well approximated by mixtures of \textit{polymer models} with convergent cluster expansions in the relevant range of parameters.  These polymer models each represent deviations from  one of a collection of \textit{ground states}. For example, in the case of the Potts model, the ground states are the monochromatic configurations, while in the case of the hard-core model the ground states are the two collections of independent sets with no occupied even or odd vertices respectively.

The main steps in the proofs  of Theorems~\ref{thmhard-core} and~\ref{thm:pottspartition} are as follows.  
\begin{enumerate}
\item First we show that the partition function of the relevant model on (bipartite) expanders is dominated by configurations that are `close' to one of the ground states. 
\item For each ground state we define a polymer model representing deviations from the given state. We show that $\alpha$-expansion implies a strong upper bound on the polymer weights, which allows us to verify the Koteck\'y-Preiss condition for the convergence of the cluster expansion~\cite{kotecky1986cluster}. 
\item This last step allows us to implement a version of the approximate counting algorithm from~\cite{helmuth2018contours}, based on truncating the cluster expansion.  This algorithm is inspired by Barvinok's method of truncating the Taylor series of the log partition function, but here we can work directly with the cluster expansion and avoid any use of complex analysis. 
\item The sampling algorithm is based on a form of self-reducibility for abstract polymer models.
\end{enumerate}

In Section~\ref{secPolymermodels} we define polymer models and the cluster expansion, we state the Koteck\'y-Preiss condition,  and describe the counting algorithm.  We prove our results for the Potts model in Section~\ref{secPottsproof},  for the hard-core model in Section~\ref{secHCproof}, and for proper colorings in Section~\ref{secColorProof}.  We conclude with some discussion and open problems in Section~\ref{secconclude}.

\section{Polymer models} 
\label{secPolymermodels}

\subsection{Abstract polymer models}
We define polymer models in sufficient generality for the purposes of this paper.  A more general treatment can be found, for example, in~\cite{gruber1971general,kotecky1986cluster}.

A \textit{polymer} $\gamma$ is a connected subgraph of $G$. A polymer model consists of a set of allowed polymers $\cC(G)$ along with a complex-valued weight function $w_\gamma$ for each polymer.  We measure the size of a polymer by $| \gamma|$,  the number of vertices of $ \gamma$.

We say two polymers $\gamma, \gamma'$ are \textit{compatible} if $d( \gamma,  \gamma') >1$ and \textit{incompatible} otherwise, where $d(\cdot, \cdot)$ is the graph distance.  
Let $\cG(G)$ be the collection of all finite subsets (including the empty set) of $\cC(G)$ consisting of mutually compatible polymers.

We can then define the polymer model partition function
\begin{align}
\Xi(G):=\sum_{\Gamma\in\cG(G)}\prod_{\gamma\in\Gamma} w_{\gamma}\, .
\end{align}

The prototypical example of a polymer model is the low-fugacity hard-core model on a graph $G$: the set of polymers $\cC(G)$ is simply the set of vertices $V(G)$.  The collection of sets of mutually compatible polymers $\cG(G)$ is exactly $\cI(G)$, the collection of independent sets of $G$.  If we set the weight function of every polymer to be $w_{\gamma} = \lam$, then the abstract polymer partition  function $\Xi(G)$ is exactly the hard-core partition function $Z_G(\lam)$.       

\subsection{Convergent cluster expansions}

A detailed probabilistic understanding of a polymer model can be obtained by showing that the \textit{cluster expansion} of its log partition function converges.

For a multiset of polymers $\Gamma$, the \textit{incompatibility graph} $H(\Gamma)$ has one vertex for each polymer in the multiset (with multiplicity) with an edge between two vertices corresponding to polymers $\gamma, \gamma'$ that are incompatible.  A \textit{cluster} is an ordered list of polymers from $\cC(G)$ (with repetitions allowed) whose incompatibility graph is connected. The size of a cluster is $|\Gamma| = \sum_{\gamma \in \Gamma} |\gamma|$.   Let $\cG^{\text{clust}}(G) $ be the collection of all clusters.   
The cluster expansion is then the (formal) power series in the variables $w_\gamma$, $\gamma \in \cC(G)$,
\begin{align}
\label{eqClusterExpansion}
\log \Xi(G) &=  \sum_{\Gamma \in \cG^{\text{clust}}(G)} \phi( \Gamma) \prod_{\gamma \in \Gamma} w_{\gamma} \,,
\end{align}
where $\phi(\Gam)$ is the \textit{Ursell function} of the incompatibility graph $H=H(\Gamma)$ defined by
\begin{align*}
\phi(H) &=  \frac{1}{|V(H)|!}\sum_{\substack{ A\subseteq E(H) \\ \text{spanning, connected}}} (-1)^{|A|} \,.
\end{align*}
(The Ursell function $\phi(H)$ is an evaluation of the Tutte polynomial of $H$). In fact the cluster expansion is simply the multivariate Taylor series for $\log \Xi(G)$ in the variables $w_\gamma$, as observed by Dobrushin~\cite{dobrushin1996estimates}.  See also Scott and Sokal~\cite{scott2005repulsive} for a derivation of the cluster expansion and much more. 

A sufficient condition for the convergence of the cluster expansion is given by the following specialization of a result of Koteck\'y and Preiss.   

\begin{theorem}[\cite{kotecky1986cluster}]\label{thmKPsimple}
Fix a function $g : \cC(G) \to [0, \infty)$, and extend $g$ to clusters by defining
\begin{align*}
g(\Gamma) &= \sum_{\gamma \in \Gamma} g(\gamma) \,.
\end{align*}
 Suppose that for all $\gamma \in \cC(G)$
\begin{equation}\label{eqKPsimple}
\sum_{\gamma':d(\gamma', \gamma)\leq1} |w_{\gamma'}| e^{|  \gamma'| + g(\gamma')} \le |  \gamma | \,.
\end{equation}
Then the cluster expansion converges absolutely and, moreover, for every $v \in V(G)$,
\begin{align}
\label{eqpinnedEstimate}
\sum_{\substack{\Gamma \in \cG^{\text{clust}}(G)\\ \Gamma \ni v}} \left | \phi(\Gamma) \prod_{\gamma \in \Gamma} w_\gamma   \right| e^{g(\Gamma)} \le 1 \,,
\end{align}
where we write $v \in \Gamma$ if there exists $\gamma \in \Gamma$ so that $v \in  \gamma$. 
\end{theorem}

We remark that the theorem appearing in 
\cite{kotecky1986cluster} is more general.
For instance it allows a general choice of function 
$f : \cC(G) \to [0, \infty)$ in lieu of the function $|\cdot|$ appearing in \eqref{eqKPsimple}.
A more careful choice of the function $f$ might lead to better dependencies between the parameters in our results, although as far as we can tell only constant factor improvements would be gained.

\subsection{Algorithms}

We can use Theorem~\ref{thmKPsimple} to approximate the partition function.  For a given function $g : \cC(G) \to [0, \infty)$ for which~\eqref{eqKPsimple} holds, define the truncated cluster expansion
\begin{align*}
T_m(G) &= \sum_{\substack{\Gamma \in \cG^{\text{clust}}(G)\\ g(\Gamma)< m}} \phi(\Gamma) \prod_{\gamma \in \Gamma} w_{\gamma} \,.
\end{align*}
We refer to $g$ as a \emph{decay function} for the polymer model.
Under the Kotecky-Preiss condition, truncating the cluster expansion with large enough $m$ gives a good approximation to the log partition function.  
\begin{fact}
\label{factKP}
If condition~\eqref{eqKPsimple} holds and $m \ge \log (|V(G)|/\eps)$,  then
\begin{align*}
\left |\log \Xi(G) -  T_m(G) \right| \le \eps \,.
\end{align*}
\end{fact}

This fact follows from summing \eqref{eqpinnedEstimate} over all $v \in V(G)$.

Under  mild conditions on the polymer model, and under the non-trivial condition of zero-freeness of the partition function in a disc in the complex plane, Helmuth, Perkins, and Regts gave an efficient approximation algorithm for the partition function \cite[Theorem 6]{helmuth2018contours}.  The following theorem is an adaptation of that theorem.
\begin{theorem}
\label{thmPolymerCount}
Fix $\Delta$ and let $\mathfrak G$ be  some class of graphs of maximum degree at most $\Delta$.
Suppose the following hold for a given polymer model with decay function $g ( \cdot)$:
\begin{enumerate}
\item[(i)] There exists constants  $c_1,c_2>0$ such that given a connected subgraph $ \gamma$, determining  whether $\gamma \in \cC(G)$ and then computing $w_{\gamma}$ and $g(\gamma)$ can be done in time $O\left(|\gamma|^{c_1} e^{c_2|  \gamma|} \right)$.  \label{item:surface}
\item[(ii)] There exists $\rho(\Delta)=\rho>0$ so that for every $G \in \mathfrak G$ and every $\gamma \in \cC(G)$, $g(\gamma) \ge \rho | \gamma|$.  \label{item:Peierl}
\item[(iii)] The Koteck\'y-Preiss condition~\eqref{eqKPsimple} holds with the given function $g( \cdot)$.  \label{item:zerofree}
\end{enumerate}
Then there is an FPTAS for $\Xi(G)$ for $G \in \mathfrak G$.  The running time of this algorithm is $n \cdot (n/\eps)^{ O( (\log \Delta+c_2)/ \rho) }  $. 
\end{theorem}

The algorithm is as follows.  Given $\eps>0$, and a graph $G$ on $n$ vertices,  let $m = \log(n/\eps)$ and:
\begin{enumerate}
\item Enumerate all clusters $\Gamma \in  \cG^{\text{clust}}(G)$ with $|\Gamma| < m/\rho$; call the list of such clusters $\cL$.
\item For each cluster $\Gamma \in \cL$ compute $\phi(\Gamma)$, $\prod_{\gamma \in \Gamma} w_{\gamma}$, and $g(\Gamma)$.
\item Compute $T_m(G)$ by summing:
\begin{align*}
T_m(G) &= \sum_{\substack{ \Gamma \in \cL \\ g(\Gamma) < m   } } \phi(\Gamma) \prod_{\gamma \in \Gamma} w_\gamma \,.
\end{align*}

\item Output $\exp(T_m(G))$. 
\end{enumerate}

\begin{proof}[Proof of Theorem~\ref{thmPolymerCount}]
Fact~\ref{factKP} along with condition (iii) of the theorem tells us that $\exp(T_m(G))$ is an $\eps$-relative approximation to $\Xi(G)$ and so we just need to show that steps 1--3 can be done in time $n \cdot (n/\eps)^{ O(( \log \Delta+c_2)/ \rho) }   $. 

By condition (ii) of the theorem, we need to list clusters $\Gamma$ with $| \Gamma | < m/\rho$, compute the Ursell function for each cluster, and compute the product of polymer weights in each cluster.  Using the algorithm of~\cite[Theorem 6]{helmuth2018contours}, we can list all clusters of size at most $m/\rho$ in time $n \Delta^{O(m/\rho)}$ and compute their associated Ursell functions (this result relies of the algorithms of~\cite{patel2019computing} for enumerating connected subgraphs in a bounded-degree graph and~\cite{bjorklund2008computing} for evaluating the Tutte polynomial, and hence the Ursell function, of a graph in time exponential in its number of vertices).   Computing $\prod w_{\gamma}$  for a single cluster can be done in time $(m/\rho)^{c_1} e^{O(c_2 m/\rho)} $ using condition (i).  Putting this together yields an algorithm with running time $n \cdot (n/\eps)^{ O( (\log \Delta+c_2)/ \rho) }  $.
\end{proof}

We can also sample efficiently from a polymer model satisfying condition~\eqref{eqKPsimple}.  We define a probability measure  $\nu_{G}$ on $\cG(G)$:
\begin{align}
\nu_{G} (\Gamma) &=  \frac{ \prod_{\gamma\in\Gamma}w_\gamma  }{ \Xi(G)    },   \quad \Gamma \in \cG(G) \,.
\end{align}

Following~\cite[Theorem 10]{helmuth2018contours}, the approximate counting algorithm also provides a polynomial-time sampling algorithm for $\nu_G$. 
\begin{theorem}[\cite{helmuth2018contours}, Theorem 10]
\label{thmPolymerSample}
Under the conditions of Theorem~\ref{thmPolymerCount}, there is a polynomial-time sampling algorithm for $\nu_{G}$ for all $G \in \mathfrak G$. 
\end{theorem}
Theorem~\ref{thmPolymerSample} is an application of self-reducibility to the abstract polymer model.  We briefly summarize the algorithm here for completeness and since we will use some of its properties later.

Let $n = |V(G)|$, and given $\eps>0$, let $m = \log(2n/\eps)$.  Let $\cC^{(m)}(G)$ be the set of all polymers $\gamma$ so that $g(\gamma)< m$, and let $\cG^{(m)}(G)$ be the collection of pairwise compatible sets of polymers from $\cC^{(m)}(G)$.  By Fact~\ref{factKP}, with probability at least $1- \eps/2$, $\Gamma$ drawn from $\nu_G$ belongs to $\cG^{(m)}(G)$, and so it suffices to sample from $\cG^{(m)}(G)$. The algorithm builds such a configuration one polymer at a time.  

Fix an arbitrary ordering of the vertices of $G$, $v_1, \dots , v_n$.   Let $\Gamma_0 = \emptyset$, and let $\cC_1 = \cC^{(m)}(G)$.  Then for $j=1, \dots , n$ do the following:
\begin{enumerate}
\item List all polymers $\gamma \in \cC_j$ so that $v_j \in  \gamma$; call this list $\cL_j$. 
\item For each $\gamma \in \cL_j$, approximate the probability $\Pr[\gamma \in \Gamma | \Gamma_{j-1} \subseteq \Gamma]$.    
\item Choose at most one polymer from $\cL_j$ according to these approximate probabilities and add it to $\Gamma_{j-1}$ to form $\Gamma_j$;  if no polymer is chosen let $\Gamma_j  = \Gamma_{j-1}$. 
\item Let $\cC_{j+1}= \cC_j \setminus \cL_j$.  
\end{enumerate}
Output $\Gamma = \Gamma_{n}$. 

 The algorithm is simply self-reducibility applied to the polymer model.  Its correctness relies on two facts: 1) at most one polymer from $\cL_j$ can appear in $\Gamma$ since any two polymers in $\cL_j$ are incompatible with each other 2) for each polymer $\gamma \in \cC^{(m)}(G)$ there is exactly one step at which it can potentially be added to $\Gamma$: the step $j$ corresponding to the first vertex $v_j$ in $\gamma$ according to the arbitrary vertex ordering. The efficiency of the sampling algorithm relies on the fact that $\cL_j$ can be formed efficiently and that the conditional probabilities can be computed efficiently.  The latter is done by approximating partition functions of polymer models with various subsets of polymers.  Crucially, if condition~\eqref{eqKPsimple} holds for $\cC(G)$ then it holds for any subset $\cC' \subset \cC(G)$.

While the algorithm of Theorem~\ref{thmPolymerCount} does not use complex analysis or refer the zeros of a partition function, it is closely related to and inspired by the approach of Barvinok~\cite{barvinok2017combinatorics} for a broad range of approximation problems that involves truncating the Taylor series of the log partition function and using the absence of zeros of the partition function in the complex plane to deduce convergence.  The cluster expansion itself is the multivariate Taylor series of $\log  \Xi$ and the fact that it is supported on clusters which are connected objects is related to the method for efficient computation of the coefficients of graph polynomials in~\cite{patel2016deterministic}.  The algorithm of~\cite{helmuth2018contours} for approximating a polymer model partition function uses the cluster expansion in an indirect way: the cluster expansion is used to prove that the partition function does not vanish in a disc around $0$ in the complex plane; then Barvinok's algorithm of truncating the univariate Taylor series is applied.  Here we use the cluster expansion directly, truncating it to approximate the log partition function, and using the guarantees of~\eqref{eqpinnedEstimate} to bound the approximation error.  The technical reason zero-freeness and the Taylor series were needed in~\cite{helmuth2018contours} is that the weight functions of the more complicated contour models used there are ratios of partition functions, instead of the explicit polymer weights used here.

To prove Theorems~\ref{thmhard-core}, \ref{thmHCrandombip}, \ref{thm:pottspartition}, and~\ref{thmColor} we will take $\mathfrak G$ to be the class of (bipartite) $\alpha$-expander graphs of maximum degree $\Delta$.  We will show that the hard-core and Potts partition functions, at sufficiently high fugacity and low temperature respectively, and the number of proper $q$-colorings for $q$ sufficiently small, can be approximated well by sums of  partition functions of abstract polymer models. This involves showing that the expansion condition implies that almost all the weight of the given partition function comes from configurations close to one of a small number of ground states.

 We then verify in each case conditions (i) of Theorem~\ref{thmPolymerCount}, which is straightforward.  

Conditions (ii) and (iii), choosing the function $g$ and showing the Koteck\'y-Preiss condition holds,  are non-trivial, and it is here that the dependence of the temperature and fugacity on the degree and expansion is determined.

We note that although we are working with low-temperature models, we are able to use the polymer model formulation instead of the more complex \textit{contour model} formulation of Pirogov-Sinai theory used for the algorithms on $\Z^d$ in~\cite{helmuth2018contours}.   The reason polymer models suffice is that the strong expansion condition allows us to express our partition functions in terms of deviations from the ground states directly and not in the recursive fashion of a contour model.

\section{The Potts model}
\label{secPottsproof}

\subsection{Approximation by a polymer model}
In this section we show that the Potts model partition function of an expander graph can be well-approximated by the partition function of a certain polymer model.
Recall that our measure of approximation is the following.
\begin{defn}
Let $Z$ be a real number. We call $\hat Z$ an \emph{$\eps$-relative approximation} to $Z$ if 
\[
e^{-\eps} \hat Z\leq Z\leq e^{\eps} \hat Z\, .
\]
\end{defn}

Given a graph $G$ and a set $S\subseteq V(G)$, recall that we let $\p S$ denote the set of vertices in $S^c$ (the complement of $S$ in $V(G)$) adjacent to a vertex in $S$ and we let $\p_e(S)$ denote the set of edges in $G$ with one endpoint in $S$ and the other in $S^c$.  We let $G[S]$ denote the subgraph of $G$ induced by the vertex set $S$.
Let $S^+=S\cup \p S$ and let $\nabla(S) = E(G[S]) \cup \partial_e(S)$, the set of edges of $G$ that are incident to a vertex in $S$. Recall our notion of expansion. 
\begin{defn}
A graph $G$ is an $\alpha$-expander if $|\p_e(S)|\geq \alpha |S|$ for all subsets $S\subseteq V(G)$ with $|S|\le |V(G)|/2$. 
\end{defn}
Let $\mathfrak G(\alpha,\Delta)$ denote the class of all $\alpha$-expander graphs with maximum degree at most~$\Delta$.
For the remainder of this section we fix a graph $G\in\mathfrak G(\alpha,\Delta)$ on $n$ vertices.

First we first show that the main contribution to the Potts model partition function of $G$ comes from colorings where one color dominates. To make this precise we make a few definitions. First let $\Omega=[q]$ and let  $\Omega^n$ be the set of all (not necessarily proper) colorings $\omega: V(G)\to [q]$, and recall that  $m(G,\omega)$ denotes the number of monochromatic edges of $G$ induced by $\omega$. The Potts model partition function is then 
\[
Z_{G,q}(\beta):=\sum_{\omega\in\Omega^n}e^{ \beta\cdot m(G,\omega)}\, .
\]

For $j\in[q]$, let
\[
\Omega^n_j=\{\omega\in\Omega^n: |\omega^{-1}(\{j\})|> n/2\}\, ,
\]
let
\[
Z^j_G(\beta):=\sum_{\omega\in\Omega^n_j}e^{ \beta\cdot m(G,\omega)} ,
\]
and let
\[
Z^\ast_G(\beta)=\sum_{j=1}^qZ^j_G(\beta)\, .
\]

\begin{lemma}\label{lem:starapprox}
For $\beta>2\log(eq)/\alpha$, $Z_G^\ast(\beta)$ is an $e^{-n}$-relative approximation to $Z_{G,q}(\beta)$.
\end{lemma} 
\begin{proof}
Let $\Omega^n_\ast=\bigcup_{i=1}^q\Omega^n_j$ and note that this is a disjoint union. Let $\omega\notin \Omega^n_\ast$, then for each $j\in[q]$ we have 
\[
|\omega^{-1}(\{j\})|\le \frac{n}{2}\, .
\]
Letting $S_j=\omega^{-1}(\{j\})$ it follows that $|\p_e(S_j)|\ge \alpha |S_j|$. The set $S_j$ consists of all vertices of $G$ with the color $j$ and so every edge lying between $S_j$ and $S_j^c$ is bichromatic. Summing over colors $j$ we thus have at least
\[
\frac{1}{2}\sum_{j=1}^q |\p_e(S_j)|\ge\frac{1}{2}\sum_{j=1}^q  \alpha |S_j|=\frac{\alpha n}{2}
\]
bichromatic edges and so
\[
m(G,\omega)\le e(G)-\frac{\alpha n}{2}\, .
\]
Using the crude bound $|\Omega^n\backslash\Omega^n_\ast|\le |\Omega^n| =  q^n$ we then have
\[
Z_{G,q}(\beta)-Z_G^\ast(\beta)=\sum_{\omega\notin \Omega^n_\ast}e^{\beta\cdot m(G,\omega)}\le q^n e^{\beta(e(G)-\alpha n/2)}\, ,
\]
and so
\[
\left|1-\frac{Z_G^\ast(\beta)}{Z_{G,q}(\beta)}\right|\le\frac{q^n e^{\beta(e(G)-\alpha n/2)}}{Z_{G,q}(\beta)}\le q^{n-1} e^{-\beta\alpha n/2}\le e^{-n}\, ,
\]
where for the second inequality we use the trivial lower bound $Z_{G,q}(\beta)>qe^{\beta e(G)}$.
\end{proof}

This allows us to focus on approximating $Z_G^\ast(\beta)$. Henceforth, let us fix $r\in[q]$.  We will refer to $r$ as the color `red'. By symmetry 
\begin{align*}
Z_G^\ast(\beta)=q\cdot Z_G^r(\beta) \, ,
\end{align*}
and so we may in fact focus on approximating $Z_G^r(\beta)$.

We now define a polymer model, whose partition function will serve as an approximation to 
$Z_G^r(\beta)$.

Define a polymer to be a set $\gam\subset V(G)$ such that 
$G[\gam]$ is connected and $|\gam|\le n/2$.
Following the set-up in Section~\ref{secPolymermodels}, we say that two polymers $\gam_1, \gam_2$ are compatible if $d(\gam_1, \gam_2)>1$.
We let $\cC=\cC(G)$ denote the set of all polymers of $G$ and let $\cG=\cG(G)$ denote the family of all sets of mutually compatible polymers from $\cC$. 
To each polymer $\gamma\in\cC$,
we assign the weight
\[
w_\gam \coloneqq e^{-\beta|\nabla(\gamma)|}Z_{G[\gam],q-1}(\beta)\, .
\]
This defines a polymer model with partition function
\[
\Xi(G)=\sum_{\Gam\in \cG}\prod_{\gam\in\Gam} w_\gam\, .
\]
We remark that since $\cG$ includes the empty set we have $\Xi(G)\geq1$.

\begin{lemma}\label{lempolymerpotts}
For $\beta>2\log(eq)/\alpha$
\[
\tilde Z_G(\beta):= e^{\beta\cdot e(G)}\cdot\Xi(G)
\]
is an $e^{-n}$-relative approximation to $Z^r_G(\beta)$.
\end{lemma}
\begin{proof}
For $S\subset V(G)$,
let $\Omega^n(S)$ denote the set of colorings 
$\omega\in\Omega^n$
such that $\omega(S)\subset[q]\bs\{r\}$ and $\omega(S^c)=\{r\}$.
Note that for $\omega\in\Omega^n(S)$,
we have
\begin{align}\label{eqmononabla}
m(G,\omega)=e(G)-|\nabla(S)|+m(G[S],\omega)\, .
\end{align}
Note also that if $\{\gam_1,\ldots,\gam_k\}$ are the connected components of $G[S]$ then
\begin{align}\label{eqcolmult}
e^{-\beta|\nabla(S)|}
Z_{G[S],q-1}(\beta)=
\prod_{i\in[k]}e^{-\beta|\nabla(\gam_i)|}
Z_{G[\gam_i],q-1}(\beta)\, .
\end{align}
We call a subset $S\subset V(G)$ \emph{small} if $|S|\le n/2$ and \emph{large} otherwise.
We call $S$ \emph{sparse} if each of the connected components of $G[S]$ is small.
Note that there is a one to one correspondence between sparse subsets of $V(G)$ and elements of $\cG$. We thus have by~\eqref{eqmononabla} and~\eqref{eqcolmult}  that
\begin{align}
\tilde Z_G(\beta)&=
e^{\beta\cdot e(G)} 
\sum_{\Gam\in \cG}
\prod_{\gam\in\Gam}
e^{-\beta|\nabla(\gamma)|}
Z_{G[\gam],q-1}(\beta)\ ,\\
&= e^{\beta\cdot e(G)} 
\sum_{S\text{ sparse}}
e^{-\beta|\nabla(S)|}
Z_{G[S],q-1}(\beta)\, ,\\
&=\sum_{S\text{ sparse}}\sum_{\omega\in \Omega^n(S)}e^{\beta\cdot m(G,\omega)}\, .
\end{align}
On the other hand 
\begin{align}
Z^r_G(\beta)=
\sum_{S\text{ small}}\sum_{\omega\in \Omega^n(S)}e^{\beta\cdot m(G,\omega)}\, ,
\end{align}
and so 
\begin{align}
\tilde Z_G(\beta)-Z^r_G(\beta)&=
\sum_{\substack{
S\text{ sparse,}\\ \text{large}}}\sum_{\omega\in \Omega^n(S)}e^{\beta\cdot m(G,\omega)}.
\end{align}
Observe that if $S$ is large and sparse then $|\p_e(S)|\ge\alpha|S|\ge\alpha n/2$, and so
by~\eqref{eqmononabla} $m(G,\omega)\le e(G)-\alpha n/2$ for all $\omega\in \Omega^n(S)$.
It follows that 
\begin{align}
\left|1- \frac{\tilde Z_G(\beta)}{Z^r_G(\beta)}\right|\le
e^{-\beta\cdot e(G)}\left(\tilde Z_G(\beta)-Z^r_G(\beta)\right)\le
q^ne^{-\beta\alpha n/2}
\le
e^{-n}\ ,
\end{align}
where we have the lower bound $\tilde Z_G(\beta)\ge e^{\beta\cdot e(G)}$ (since $\Xi(G)\geq1$) and the crude upper bound of $q^n$ on the number of $\omega$ such that $\omega\in\Omega^n(S)$ for some large and sparse set $S$.
\end{proof}

Our aim is to apply Theorem~\ref{thmPolymerCount} to this polymer model in order to obtain an FPTAS for $\Xi(G)$ when $\beta$ is sufficiently large which, by Lemmas~\ref{lem:starapprox} and \ref{lempolymerpotts}, will furnish us with an FPTAS for $Z_{G,q}(\beta)$.

To this end we aim to verify conditions (i)--(iii) of Theorem~\ref{thmPolymerCount} for our polymer model. Verifying condition (i) is essentially immediate. Given $\gamma$, determining  whether $\gamma \in \cC$ amounts to checking whether $G[ \gamma]$ is connected and of size at most $n/2$. This can be done in $O(|  \gamma|)$ time by a depth-first search algorithm. Computing the weight function $w_{\gamma}$ can be done in time $O(\Delta|\gam|(q-1)^{|\gamma|})$ by calculating $m(G[\gam],\omega)$ for all assignments $\omega$ of $(q-1)$ colors to $\gamma$. 

We now turn our attention to verifying conditions (ii) and (iii) for an appropriate choice of function $g$.

\subsection{Verifying the Koteck\'y-Preiss condition}
We choose 
\[ g(\gamma) = |\gamma| \,.\]
Condition (ii) of Theorem~\ref{thmPolymerCount} holds trivially with $\rho=1$. 
It remains to show the Koteck\'y-Preiss condition holds.  That is, 
\begin{align*}\label{eqKPsimple}
\sum_{\gamma':d(\gamma', \gamma)\leq1} w_{\gamma'} e^{2|  \gamma'|}\le |  \gamma | 
\end{align*}
for all $\gamma\in \cC(G)$. 
We begin by bounding $w_{\gamma}$:
\begin{align*}
w_\gamma &= e^{-\beta|\nabla(\gamma)|}Z_{G[\gam],q-1}(\beta) \\
&\le e^{-\beta  |\partial_e \gamma| } (q-1)^{|\gamma|} \\
&\le e^{-\beta \alpha |\gamma| } (q-1)^{|\gamma|}\,.
\end{align*}
For the first inequality we used that $|\nabla(\gam)|=|\p_e\gam|+e(G[\gam])$ and that 
$Z_{G[\gam],q-1}(\beta)\leq (q-1)^{|\gam|}e^{\beta\cdot e(G[\gam])}$. 
For the second inequality we used that $G$ is an $\alpha$-expander.
It thus suffices to show that
\begin{align*}
\sum_{\gamma':d(\gamma',\gamma)\le1} e^{ (2- \alpha \beta    + \log (q-1) ) |\gamma'|  }  &\le  |\gamma| \,.
\end{align*}
If we could show that for each $v\in V(G)$
\begin{equation}\label{eq:vertex}
\sum_{\gamma':\gamma'\ni v}e^{(2-\alpha \beta  +\log(q-1))|\gamma'|}\le \frac{1}{\Delta+1}\, ,
\end{equation}
then by summing this inequality over all $v\in \gamma^+(=\gamma\cup \p \gamma$) and noting that $| \gamma^+|\le(\Delta+1)|\gamma|$, we would be done.

In order to establish \eqref{eq:vertex} we borrow the following lemma:
\begin{lemma}[\cite{galvin2004phase}, Lemma 2.1]\label{lemConCount}
In a graph of maximum degree at most $\Delta$, the number of connected, induced subgraphs of order $t$ containing a fixed vertex $v$ is at most $(e\Delta)^t$.
\end{lemma}

It follows that the number of polymers $ \gamma$ of size $t$ that contain a given vertex $v$ is bounded by $(e\Delta)^t$, and so
\[
\sum_{\gamma':\gamma'\ni v}e^{(2-\alpha \beta + \log(q-1))|\gamma'|}\le\sum_{t=1}^\infty  \left((q-1)\Delta\cdot e^{(3-\alpha\beta)}\right)^t \le \frac{1}{\Delta+1}\, ,
\]
for $\beta \ge  \frac{ 4+   2\log (q \Delta)   }{ \alpha  }$. 

This verifies condition (iii) of Theorem~\ref{thmPolymerCount} 
 and so Theorem~\ref{thmPolymerCount} gives an FPTAS for $\Xi(G)$ for all $\beta \ge  \frac{ 4+   2\log (q \Delta)  }{ \alpha  }$.    

\subsection{Proof of Theorem~\ref{thm:pottspartition}}

We consider two cases separately.  If $\eps \le e^{-n/2}$, then we proceed by brute force, calculating $m(G,\omega)$ for each of the $q^n$ possible colorings of $G$. In this way we can calculate the partition function $Z_{G,q}(\beta)$ exactly in time $O(n \Delta q^n)$ and therefore in time polynomial in $1/\eps$. Similarly we can obtain an exact sample from $\mu_{G,q,\beta}$ by brute force in time polynomial in $1/\eps$.

Now we assume $\eps > e^{-n/2}$. 
Using Theorem~\ref{thmPolymerCount}, we can obtain $Z_{\text{alg}}$, an $\eps/2$-relative approximation to $e^{\beta e(G)} \cdot  \Xi(G)$ in time polynomial in $n$ and $1/\eps$ (here we use the fact that for a polymer $\gamma$, $w_\gamma$ can be computed  in time $O(|\gamma| \Delta q^{|\gamma|})$).
By Lemma~\ref{lempolymerpotts}, $e^{\beta e(G)} \cdot  \Xi(G)$ is an $e^{-n}$-relative approximation to $Z_G^r(\beta)$ and so $qe^{\beta e(G)} \cdot  \Xi(G)$ is an $e^{-n}$-relative approximation to $Z_G^\ast(\beta)$ . 
By Lemma~\ref{lem:starapprox}, it follows that $qe^{\beta e(G)} \cdot  \Xi(G)$ is an $\eps/2$-relative approximation to $Z_{G,q}(\beta)$ and so $qZ_{\text{alg}}$ is an $\eps$-relative approximation to $Z_{G,q}(\beta)$ as required.  Lemmas~\ref{lem:starapprox} and~\ref{lempolymerpotts} apply since $  \frac{ 4+   2\log (q\Delta)    }{ \alpha  } > 2 \log(eq )/\alpha$.

For the approximate sampling algorithm, we will apply Theorem~\ref{thmPolymerSample}.  

Consider the following distribution $\hat \mu$ on $\Omega^n$.  Choose $r \in [q]$ uniformly at random and  then sample $\Gamma$ from the measure  $ \nu_{G,\beta}$ on  $\cG(G)$ defined by
\begin{align*}
 \nu_{G,\beta}(\Gamma ) & = \frac{ \prod_{\gamma \in \Gamma}  w_{\gamma}   }{ \Xi(G)}   \,.
\end{align*}
Then convert $\Gamma$ into a $q$-coloring $\omega \in \Omega^n$ as follows: for each $\gamma \in \Gamma$, choose a $(q-1)$ coloring (excluding color $r$)  $\omega_\gamma$ according to $\mu_{G[\gamma],q-1,\beta}$.   Then set $\omega(v) = \omega_{ \gamma}(v)$ if $v \in \gamma$ for some $\gamma \in \Gamma$, and $\omega(v) = r$ otherwise. The resulting distribution of $\omega$ is $\hat \mu$.

By Lemmas~\ref{lem:starapprox} and~\ref{lempolymerpotts} and the condition on $\eps$,
 \[ \| \hat \mu -  \mu_{G,q,\beta}\|_{TV} \le \eps /2 . \] 
Thus to complete the proof of Theorem~\ref{thm:pottspartition} it suffices to obtain an $\eps/2$-approximate sample from $ \nu^r_{G,\beta}$ in time polynomial in $n$ and $1/\eps$, and to do this we appeal to Theorem~\ref{thmPolymerSample}.
 Crucially, because of the truncation involved in the polymer sampling algorithm of Theorem~\ref{thmPolymerSample}, we have $|\gamma| = O( \log(n/\eps))$ for each $\gamma \in \Gamma$ and so for each $\gamma$ a sample from $\mu_{G[\gamma],q-1,\beta}$ can be obtained by brute force in time polynomial in $n$ and $1/\eps$.

\qed

\subsection{Proof of Corollary~\ref{corPottsrandom}}
In order to prove Corollary~\ref{corPottsrandom} we require a very brief review of spectral graph theory.
For a graph $G$ on $n$ vertices we let
\[
h(G):=\min_{\{S: |S|\leq n/2\}}\frac{|\p_eS|}{|S|}\, .
\]
In other words, $h(G)$ is the largest value of $\alpha$ for which $G$ is an $\alpha$-expander.
The exact determination of $h(G)$, given G, is known to be co\textbf{NP}-complete. 
However, as shown by Alon~\cite{alon1986eigenvalues}, one can efficiently approximate the expansion properties of a graph using its \emph{spectrum}.
For a graph $G$ on $n$ vertices, we let $\lam_1 \geq\ldots\geq\lam_n$ denote the eigenvalues of the adjacency matrix of $G$ 
(that is, the $n \times n$ matrix $A$, with rows and columns indexed by $V(G)$, where $A_{ij}=1$ if $\{i,j\}\in E(G)$ and $A_{ij}=0$ otherwise).

The following lemma, one of Cheeger's inequalities (see e.g. \cite{hoory2006expander}), bounds $h(G)$ in terms of the second largest eigenvalue of $G$ when $G$ is regular.
\begin{lemma}\label{lemh}
If $G$ is a $\Delta$-regular graph, then
\[
h(G)\geq\frac{\Delta-\lam_2}{2}\, .
\]
\end{lemma}
We will also use the following celebrated result of Friedman on the spectral gap of the random regular graph.
For a graph $G$ 
we let $\lam(G)=\max(|\lam_2|, |\lam_n|)$.

\begin{theorem}[Friedman~\cite{friedman2003proof}] \label{thmfriedman}
For every fixed $\eps>0$, almost every $\Delta$-regular graph $G$ satisfies
\begin{align}\label{eqspec}
\lam(G)\leq 2\sqrt{\Delta-1}+\eps\, .
\end{align}
\end{theorem}

\begin{proof}[Proof of Corollary~\ref{corPottsrandom}]
Fix $\eps=1/100$. Since the eigenvalues of a graph $G$ can be calculated in polynomial time, there is a polynomial time algorithm to determine whether a graph $G$ satisfies \eqref{eqspec}. 
If $G$ satisfies  \eqref{eqspec} then by Lemma~\ref{lemh}, $G$ is a  $\Delta/40$-expander (since for $\Delta \ge 3$, $\Delta -2\sqrt{\Delta-1} -1/100\ge \Delta/20$). By Theorem~\ref{thm:pottspartition} it follows that if $\beta > 200\frac{\log( q\Delta)}{\Delta}$, there is an FPTAS and polynomial-time sampling algorithm for the ferromagnetic Potts model at inverse temperature $\beta$ on $G$. We conclude by noting that almost every $\Delta$-regular graph satisfies \eqref{eqspec} by Theorem~\ref{thmfriedman}.
\end{proof}

\section{The hard-core model}
\label{secHCproof}
\subsection{Approximation by a polymer model}\label{subsecapprox}

In this section we prove Theorem~\ref{thmhard-core} following the same strategy as in the previous section. First we approximate the hard-core partition function by a sum of polymer model partition functions, and then we verify the conditions of Theorem~\ref{thmPolymerCount} for these models.

We let $G=(\cO,\cE, E)$ denote a bipartite graph with partition classes $\cO$, $\cE$ and edge set $E$. We will refer to vertices of $\cO$ and $\cE$ as `odd' and `even' vertices respectively.

Recall our notion of expansion for a bipartite graph.
\begin{defn}\label{defbipexpsimple}
For $\alpha >0$, a bipartite graph $G=(\cO,\cE, E)$ is a bipartite $\alpha$-expander if $|\partial S|\geq (1+\alpha) |S|$ for all $S\subseteq \cO$ with $|S|\le |\cO|/2$ and all $S\subseteq \cE$ with $|S|\le |\cE|/2$.
\end{defn}

We remark that for regular bipartite graphs the expansion property of Definition~\ref{defbipexpsimple} also follows from a spectral gap condition.
Indeed if $G$ is a $\Delta$-regular connected bipartite graph whose second largest eigenvalue is $\lam_2$, then a classical result of Tanner~\cite{tanner1984explicit} implies that $G$ is a bipartite $\alpha$-expander where $\alpha=(\Delta^2-\lam_2^2)/(\Delta^2+\lam_2^2)$.

Let  $\mathfrak G^{\text{bip}}(\alpha,\Delta)$ denote the class of all bipartite $\alpha$-expander graphs with maximum degree at most $\Delta$. From this point on, let us fix a graph $G\in\mathfrak G^{\text{bip}}(\alpha,\Delta)$ on $n$ vertices with partition classes $\cO, \cE$. Recall that $\cI(G)$ denotes the family of all independent sets in $G$.

Let us call a set $S\subseteq \cO$ \emph{small} if $|S|\le|\cO|/2$. We define small subsets of $\cE$ similarly. 

\begin{lemma}\label{lemsmall}
Let $I\in \cI(G)$, then at least one of the sets $I\cap\cO, I\cap\cE$ is small.
\end{lemma}
\begin{proof}
Without loss of generality let $|\cO|\ge|\cE|$ and suppose that $|I\cap\cO|>|\cO|/2$. Since $G$ is a bipartite $\alpha$-expander, by considering a subset of $I\cap\cO$ of size $|\cO|/2$, we have $|\partial(I\cap\cO)|\geq (1+\alpha) |\cO|/2$ and so 
\[
|I\cap\cE|<|\cE|-(1+\alpha) |\cO|/2<|\cE|/2\, .
\]

\end{proof}

We now define two distinct polymer models whose partitions functions we will use to approximate $Z_G(\lam)$. 

First let us introduce some notation and terminology. 
We let $G^k$ denote the $k$th power of the graph $G$, that is the graph on vertex set $V(G)$ where two vertices are adjacent if and only if they are at distance at most $k$ in $G$.

\begin{defn}\label{defklink}
We say a subset $S \subseteq V(G)$ is \emph{$G^k$-connected} if the induced subgraph $G^k[S]$ is connected. We call the connected components of $G^k[S]$, the \emph{$G^k$-connected} components of $S$.
\end{defn}
Throughout this section we will be concerned only with $G^2$-connected sets. Note that a set $S \subset \cO, \cE$ is $G^2$-connected if and only if $G[S^+]$ is connected.

A \textit{polymer} is any \emph{small} $G^2$-connected set $\gamma$ which lies entirely in $\cE$ or entirely in $\cO$; in the first case we say it is an even polymer and in the second an odd polymer.  The size of a polymer $\gamma$, $| \gamma|$,  is again the number of vertices in $\gamma$. We define a compatibility relation on the set of even (odd) polymers, with $\gamma_1, \gamma_2$  compatible if $d_G(\gamma_1, \gamma_2) >2$  (that is if $\gamma_1 \cup \gamma_2$ is not $G^2$-connected); otherwise $\gamma_1$ and $\gamma_2$ are incompatible.  Note that this is consistent with the definition of a polymer model in Section~\ref{secPolymermodels}: a polymer is a connected subgraph in $G^2$ and two polymers $\gamma, \gamma'$ are compatible if their graph distance in $G^2$ is greater than $1$. 

 For each polymer $\gamma$ we assign a weight
\[
w_\gamma:=\frac{\lam^{|\gamma|}}{(1+\lam)^{|\partial \gamma|}}\, .
\]

We let $\cC^{\cE}=\cC^{\cE}(G)$ denote the set of all even polymers of $G$ and let $ \cG^{\cE}=\cG^{\cE}(G)$ denote the family of all sets of mutually compatible polymers from $\cC^\cE$. We define $\cC^{\cO}$, $\mathcal G^{\cO}$ similarly.

The set of polymers $\cC^{\cE}$ ($\cC^{\cO}$) constitutes a polymer model with partition function $\Xi^{\cE}(G)$ ($\Xi^{\cO}(G)$). That is, 
\[
\Xi^{\cE}(G)=\sum_{\Gamma\in {\mathcal G^{\cE}}}\, \prod_{\gamma\in\Gamma}w_\gamma\,\, \, \text{  and  }\, \, \, \Xi^{\cO}(G)=\sum_{\Gamma\in {\mathcal G^{\cO}}}\, \prod_{\gamma\in\Gamma}w_\gamma\, .
\]
Note that $(1+\lam)^{|\cO|}\Xi^{\cE}(G)$ represents the contribution to $Z_G(\lam)$ from independent sets that are dominated by odd occupied vertices, and vice-versa.

We now show that a certain linear combination of the  partition functions $\Xi^{\cE}(G), \Xi^{\cO}(G)$ serves as a good approximation to $Z_G(\lam)$.

\begin{lemma}\label{lemhcapprox}
For $\lam>e^{11/\alpha}$, the polynomial 
\[
\tilde{Z}_G(\lam)=(1+\lam)^{|\cO|}\Xi^{\cE}(G)+(1+\lam)^{|\cE|}\Xi^{\cO}(G)
\]
is an $e^{-n}$-relative approximation to $Z_G(\lam)$.
\end{lemma}
\begin{proof}
Let us call a set $A\subseteq \cE, \cO$ \emph{sparse} if its $G^2$-connected components are all small. Note that $\mathcal G^{\cE}$ ($\mathcal G^{\cO}$) is in one to one correspondence with sparse subsets of $\cE$ ($\cO$) and so
\begin{align}
\tilde{Z}_G(\lam)&=(1+\lam)^{|\cO|}\sum_{\Gamma\in {\mathcal G^{\cE}}}\, \prod_{\gamma\in\Gamma}\frac{\lam^{|\gamma|}}{(1+\lam)^{|\partial \gamma|}}+(1+\lam)^{|\cE|}\sum_{\Gamma\in {\mathcal G^{\cO}}}\, \prod_{\gamma\in\Gamma}\frac{\lam^{|\gamma|}}{(1+\lam)^{|\partial \gamma|}}\\
&=(1+\lam)^{|\cO|}\sum_{\text{sparse }A\subseteq \cE}\frac{\lam^{|A|}}{(1+\lam)^{|\partial A|}}+(1+\lam)^{|\cE|}\sum_{\text{sparse }A\subseteq \cO}\frac{\lam^{|A|}}{(1+\lam)^{|\partial A|}}\, .\label{eqsparsedouble}
\end{align}

Consider an independent $I\in \cI(G)$ for which $I\cap \cE=A$ where $A$ is some fixed subset of $\cE$. The possible intersections $I\cap\cO$ are then precisely the subsets of $\cO\backslash \partial A$. The contribution to $Z_G(\lam)$ from independent sets $I$ such that $I\cap \cE=A$ is therefore $\lam^{|A|}(1+\lam)^{|\cO|-|\partial A|}$. Similarly the contribution to $Z_G(\lam)$ from independent sets $I$ such that $I\cap \cO=B$ is $\lam^{|B|}(1+\lam)^{|\cE|-|\partial B|}$. Let us call an independent set $I$ sparse if both $I\cap\cE, I\cap\cO$ are sparse. Since by Lemma~\ref{lemsmall}, for any independent set $I$, at least one of $I\cap\cE$, $I\cap\cO$ is small (and therefore sparse), the sums in~\eqref{eqsparsedouble} contain the contribution to $Z_G(\lam)$ from all $I\in \cI(G)$ and double count the contribution from precisely the sparse independent sets i.e.
\[
\tilde{Z}_G(\lam)=Z_G(\lam)+\sum_{I \text{ sparse}}\lam^{|I|}\, .
\] 

Let $I$ be a sparse independent set. Since $I\cap \cE$, $I\cap \cO$ are composed of small $G^2$-connected components it follows that $|\partial(I\cap\cE)|>(1+\alpha)|I\cap\cE|$ and $|\partial(I\cap\cO)|>(1+\alpha)|I\cap\cO|$. 
Since each vertex in $\partial(I\cap\cE)$ and $\partial(I\cap\cO)$ must be unoccupied, it follows that 
\[
|I\cap \cE|<|\cE|-(1+\alpha)|I\cap\cO|
\]

and 

\[
|I\cap \cO|<|\cO|-(1+\alpha)|I\cap\cE|\, .
\]

By summing these two inequalities we conclude that
\[
|I|<\frac{n}{2+\alpha}\, .
\]

It follows that
\[
\sum_{I \text{ sparse}}\lam^{|I|}\leq2^n\lam^{\tfrac{n}{2+\alpha}}\, .
\]
Using the crude bound $Z_G(\lam)\geq \lam^{n/2}$ we have
\[
\left| 1-\frac{\tilde Z_G(\lam)}{Z_G(\lam)}\right|<2^n\frac{\lam^{\tfrac{n}{2+\alpha}}}{\lam^{n/2}}= 2^n\lam^{-\tfrac{\alpha}{4+2\alpha}n}\le e^{-n}\, .
\]

\end{proof}

In order to approximate $Z_G(\lam)$ we approximate $\Xi^{\cE}(G)$ and $\Xi^{\cO}(G)$ separately.
We focus on approximating $\Xi^{\cE}(G)$, noting that the approximation algorithm for $\Xi^{\cO}(G)$ will be identical up to a change of notation.

We will verify conditions (i)--(iii) of Theorem~\ref{thmPolymerCount} for the polymer model defined on $\cC^{\cE}$ in order to obtain an FPTAS for $\Xi^{\cE}(G)$. Verifying condition (i) is essentially immediate. Given $ \gamma\subset V(G)$, determining  whether $\gamma \in \cC^\cE$ amounts to checking whether $\gamma\subset\cE$  and whether $G[\gamma^+]$ is connected. This can be done in $O(\Delta |\gamma|)$ time by a depth-first search algorithm. Computing $|\gamma|, |\partial \gamma|$ and thus $w_{\gamma}$ can also clearly be done in $O(\Delta |\gamma|)$ time. 

We now turn our attention verifying conditions (ii) and (iii) for an appropriate choice of function $g$.

\subsection{Verifying the Koteck\'y-Preiss condition}\label{subsecVKP}

We choose 
\[ g(\gamma)= |\gamma|   \,.\]
Condition (ii) of Theorem~\ref{thmPolymerCount} holds trivially with $\rho = 1$. 
It remains to show the Koteck\'y-Preiss condition holds.  That is, 
\begin{align*}\label{eqKPsimple}
\sum_{\gamma':d(\gamma', \gamma)\leq1} w_{\gamma'} e^{g(\gamma') +|  \gamma'|}\le |  \gamma | 
\end{align*}
for all $\gamma\in \cC^\cE$.
Note that since $G$ is a bipartite $\alpha$-expander, for $\gamma\in \cC^\cE$ we have
\[
w_\gamma=\frac{\lam^{|\gamma|}}{(1+\lam)^{|\partial \gamma|}}\leq \frac{\lam^{|\gamma|}}{(1+\lam)^{(1+\alpha)|\gamma|}}\leq (1+\lam)^{-\alpha|\gamma|}\, .
\]
It thus suffices to show that 

\[
\sum_{\gamma':d(\gamma', \gamma)\leq1}(1+ \lam)^{-\alpha|\gamma'|} \cdot e^{2 | \gamma'|}\le |  \gamma | \,.
\]

If we could show that for each $v\in V(G)$

\begin{equation}\label{eq:vertexhc}
\sum_{\gamma':\gamma'\ni v }(1+ \lam)^{-\alpha|\gamma'|} \cdot e^{2|\gamma'|}\le \frac{1}{\Delta^2}\, ,
\end{equation}
then by summing this inequality over all $v\in \cE$ at distance at most $2$ from $\gam$ (noting that there are at most $(\Delta(\Delta-1)+1) |\gamma|\le\Delta^2|\gamma|$ such vertices), we would be done. 

In order to establish \eqref{eq:vertexhc},
first observe that the graph $G^2$ has maximum degree at most $\Delta^2$
and so by Lemma~\ref{lemConCount} the number of $G^2$-connected
sets of size $t$ containing vertex $v$ is at most $(e\Delta^2)^t$.
We thus have 
\[
\sum_{\gamma':\gamma'\ni v } (1+\lam)^{-\alpha|\gamma'|} \cdot e^{2|\gamma'|}\leq \sum_{t=1}^\infty (e^3\Delta^2(1+\lam)^{-\alpha})^t\leq\frac{1}{\Delta^2}
\]
provided $\lam>(2e^3\Delta^4)^{1/\alpha}$.

\subsection{Proof of Theorem~\ref{thmhard-core}}
\label{secHCfinish}

We consider two cases separately.  If $\eps < 2^{-n}$, then we proceed by brute force, checking all subsets of $V(G)$ to see if they are independent. In this way we can calculate the partition function $Z_G(\lam)$ exactly in time $O(n2^n)$ and therefore count and sample in time $n/\eps$. 

Now we assume $\eps > 2^{-n}$ and take 
\begin{align}
\label{eqlambound}
\lam > \max\left\{(2e^3\Delta^4)^{1/\alpha},e^{11/\alpha}\right\} \,.
\end{align}
Assume without loss of generality that $|\cO|\ge|\cE|$. Note that by the definition of a bipartite expander we also have $|\cE|\ge |\cO|/2$. 

Using the FPTAS for $ \Xi^{\cE}(G)$ given by Theorem~\ref{thmPolymerCount}, we may find $Z^\cE_{\text{alg}}$, an $\eps/2$-relative approximation to $ (1+\lam)^{|\cO|}\Xi^{\cE}(G)$ in $(n/\eps)^{O(\log \Delta)}$-time.

In identical fashion we may find $Z^\cO_{\text{alg}}$, an $\eps/2$-relative approximation to $(1+\lam)^{|\cE|}\Xi^{\cO}(G)$ in $(n/\eps)^{O(\log \Delta)}$-time. It follows that $Z^{\text{alg}}:=Z^\cO_{\text{alg}}+Z^\cE_{\text{alg}}$ is an $\eps/2$-relative approximation to $\tilde Z_G(\lam)$ (as defined in Lemma~\ref{lemhcapprox}). By Lemma~\ref{lemhcapprox}, $\tilde Z_G(\lam)$ is an $\eps/2$-relative approximation to $Z_G(\lam)$ and so $Z^{\text{alg}}$ is an $\eps$-relative approximation to $Z_G(\lam)$ as required.

The proof  for the sampling algorithm is much like that for the Potts model except that we lack the exact symmetry between ground states.  Consider the distribution $\hat \mu$ on  $\cI(G)$ defined as follows. First choose $\cE$ or $\cO$ with probability proportional to $(1+\lam)^{|\cO|}\Xi^{\cE}(G)$ and  $(1+\lam)^{|\cE|}\Xi^{\cO}(G) $ respectively. Then, supposing we chose $\cO$, sample $\Gamma$  from the measure
\begin{align*}
\nu^\cO_{G}(\Gamma) &= \frac{ \prod_{\gamma \in \Gamma} w_\gamma }{ \Xi^\cO(G) }\, .   
\end{align*}
We then set 
\[
I=J\cup\bigcup_{ \gamma\in\Gamma} \gamma
\]
where we sample $J$ from the set $\cE\backslash \bigcup_{ \gamma\in\Gamma} \partial \gamma$ by including each vertex independently with probability $\frac{\lam}{1+\lam}$.
The distribution of $I$ is $\hat \mu$.

By Lemma~\ref{lemhcapprox} we have
\begin{align*}
\| \hat \mu - \mu_{G,\lam} \|_{TV} = O( e^{-n} ) \,
\end{align*}
 and so to obtain an $\eps$-approximate sample from $\mu_{G,\lam}$ efficiently, it suffices to obtain an $\eps/2$-approximate sample from $\hat \mu$.  We do this as follows:
 \begin{enumerate}
\item Compute $Z^\cE_{\text{alg}}$, an $\eps/8$-relative approximation to $ (1+\lam)^{|\cO|}\Xi^{\cE}(G)$, and $Z^\cO_{\text{alg}}$, an $\eps/8$-relative approximation to $ (1+\lam)^{|\cE|}\Xi^{\cO}(G)$. 
\item Choose $\cO$ with probability $\frac{Z^\cO_{\text{alg}}  }{Z^\cE_{\text{alg}}+Z^\cO_{\text{alg}}  }$ and $\cE$ otherwise.
\item Then, supposing  we chose $\cO$, take $\Gamma$, an $\eps/4$ approximate sample from  $\nu^\cO_{G}$, and let
\[
I=J\cup\bigcup_{ \gamma\in\Gamma} \gamma
\]
where we sample $J$ from the set $\cE\backslash \bigcup_{ \gamma\in\Gamma} \partial \gamma$ by including each vertex independently with probability $\frac{\lam}{1+\lam}$.
\end{enumerate}
The resulting distribution on independent sets is within $\eps/2$ total variation distance of $\hat \mu$, and we can obtain the sample in time polynomial in $n$ and $1/\eps$: the computation of $Z^\cE_{\text{alg}}$ and $Z^\cO_{\text{alg}}$ is done as above, and the approximate sample from  $\nu^\cO_{G}$ or  $\nu^\cE_{G}$ is obtained efficiently by applying Theorem~\ref{thmPolymerSample}.

\qed

\subsection{Proof of Theorem~\ref{thmHCrandombip}}
\label{secHCrandom}

To prove Theorem~\ref{thmHCrandombip} we need a result on the expansion of random regular bipartite graphs. In order to state the result we first generalize our notion of expansion slightly.

\begin{defn}\label{defbipexp}
For $\rho >0$ and $\sigma \in (0,1)$, a bipartite graph $G=(\cO,\cE, E)$ is a bipartite $(\sigma,\rho)$-expander if $|\partial S|\geq \rho |S|$ for all $S\subseteq \cO$ with $|S|\le \sigma |\cO|$ and all $S\subseteq \cE$ with $|S|\le \sigma |\cE|$.
\end{defn}

Note that our previous definition of a bipartite $\alpha$-expander is the same notion as a $(1/2, 1+\alpha)$-expander.

\begin{theorem}[Bassalygo \cite{bassalygo1981asymptotically}] \label{thm:bipexpand}
Almost every $\Delta$-regular bipartite graph is an $(\sigma, \rho)$-expander provided
\[
\Delta>\frac{H(\sigma)+H(\sigma\rho)}{H(\sigma)-\sigma\rho H(1/\rho)}\, ,
\]
where $H(p)=-p\log_2(p)-(1-p)\log_2(1-p)$ is the binary entropy function. 
\end{theorem}

We will take advantage of the fact that small sets in the random regular bipartite graph expand by a lot.
\begin{lemma}\label{lem:bassexp}
There exists $\Delta_0$ such that for all $\Delta\ge\Delta_0$, almost every $\Delta$-regular bipartite graph is a $\left(\tfrac{4\log\Delta}{\Delta}, \tfrac{\Delta}{4\log\Delta}-\tfrac{1}{2}\right)$-expander.
\end{lemma}
\begin{proof}
By Theorem~\ref{thm:bipexpand} it suffices to verify that 
\begin{align}\label{eqdeltabass}
\Delta>\frac{H(\sigma)+H(\sigma\rho)}{H(\sigma)-\sigma\rho H(1/\rho)}\, ,
\end{align}
where $\sigma=4\log(\Delta)/\Delta$ and $\rho=1/\sigma-1/2$. Using the fact 
\begin{align}\label{eqHtaylor}
H(x)=\frac{x\log(1/x)+x-x^2/2}{\log 2}+O(x^3)\, ,
\end{align}
(noting that this expansion is also valid for $H(1-x)$ by symmetry of the entropy function)
 we have 
\[
H(\sigma)+H(\sigma\rho)= \frac{6}{\log 2}\cdot\frac{\log^2\Delta}{\Delta}+O\left(\frac{\log \Delta}{\Delta}\right)\, 
\]
and
\[
H(\sigma)-\sigma\rho H(1/\rho)=\frac{8}{\log 2}\cdot\frac{\log^{2} \Delta}{\Delta^2}+O\left(\frac{\log^3 \Delta}{\Delta^3}\right)\, .
\]
It follows that there exists $\Delta_0$ such that 
\eqref{eqdeltabass} holds for all $\Delta\geq\Delta_0$.
\end{proof}

Henceforth we will assume that $G=(\cO, \cE, E)$ is a $\Delta$-regular $\left(\tfrac{4\log\Delta}{\Delta}, \tfrac{\Delta}{4\log\Delta}-\tfrac{1}{2}\right)$-expander on $n=2m$ vertices. This strong expansion condition allows us to prove a strengthened version of Lemma~\ref{lemsmall}. Let us update our notion of a small set and say that a set $S\subseteq V(G)$ is \emph{tiny} if $|S|\leq \tfrac{4\log\Delta}{\Delta} m$.

\begin{lemma}\label{lemsmallrandom}
Let $I\in \cI(G)$, then at least one of the sets $I\cap\cO, I\cap\cE$ is tiny.
\end{lemma}
\begin{proof}
Suppose that $|I\cap\cE|\ge \frac{4\log\Delta}{\Delta}m$, then by the expansion property

\[
|\partial (I\cap\cE)|\ge\left(1-\frac{2\log\Delta}{\Delta}\right)m
\]
and so $|I\cap\cO|\le\tfrac{2\log\Delta}{\Delta}m$. 
\end{proof}

We define two polymer models in identical fashion to Section~\ref{subsecapprox}:  An even (odd) \textit{polymer} is any \emph{tiny} $G^2$-connected set $\gamma$ which lies entirely in $\cE$ ($\cO$). We say two even (odd) polymers $\gamma_1, \gamma_2$ are \textit{compatible} if $d_G(\gamma_1, \gamma_2) >2$ and for each polymer $\gamma$ we assign a weight
\[
w_\gamma:=\frac{\lam^{|\gamma|}}{(1+\lam)^{|\partial \gamma|}}\, .
\]

We let $\cC^{\cE}$ denote the set of all even polymers of $G$ and let $\mathcal G^{\cE}$ denote the family of all sets of mutually compatible polymers from $\cC^\cE$. We define $\cC^{\cO}$, $\mathcal G^{\cO}$ similarly. Let us again denote the partition functions associated to these polymer models by $\Xi^{\cE}(G), \Xi^{\cO}(G)$ respectively. As before a linear combination of these two partition functions serves a good approximation to $Z_G(\lam)$.

\begin{lemma}\label{lemhcapproxrandom}
There exists $\Delta_0$ such that for $\Delta\ge\Delta_0$, $\lam>20\tfrac{\log^2\Delta}{\Delta}$, the polynomial 
\[
\tilde{Z}_G(\lam)=(1+\lam)^{m}\left(\Xi^{\cE}(G)+\Xi^{\cO}(G)\right)
\]
is a $(1+\lam)^{-n/4}$-relative approximation to $Z_G(\lam)$.
\end{lemma}
\begin{proof}
As in the proof of Lemma~\ref{lemhcapprox}, we have 
\begin{equation}\label{eqsparsesum}
\tilde{Z}_G(\lam)=Z_G(\lam)+\sum_{I \text{ sparse}}\lam^{|I|}\, ,
\end{equation}
where we call an independent set $I$ \emph{sparse} if the $G^2$-connected components of $I\cap\cE$ and $I\cap\cO$ are all tiny.

Let $I$ be a sparse independent set. Since $I\cap \cE$, $I\cap \cO$ are composed of tiny $G^2$-connected components it follows that $|\partial (I\cap\cE)|> \left(\tfrac{\Delta}{4\log\Delta}-\tfrac{1}{2}\right)|I\cap\cE|$ and $|\partial (I\cap\cO)|> \left(\tfrac{\Delta}{4\log\Delta}-\tfrac{1}{2}\right)|I\cap\cO|$. 
Since each element of $\partial(I\cap\cE)$ and $\partial(I\cap\cO)$ must be unoccupied, it follows that 
\[
|I\cap \cE|<m- \left(\frac{\Delta}{4\log\Delta}-\frac{1}{2}\right)|I\cap\cO|
\]

and 

\[
|I\cap \cO|<m- \left(\frac{\Delta}{4\log\Delta}-\frac{1}{2}\right)|I\cap\cE|\, .
\]

By summing these two inequalities we conclude that
\begin{equation}\label{eqsparse}
|I|<\frac{n}{\tfrac{\Delta}{4\log\Delta}+\tfrac{1}{2}}\, .
\end{equation}

We use the following well-known estimate
\begin{equation}\label{eq:entropy}
\binom{n}{\le cn}:=\sum_{i=0}^{\lfloor cn \rfloor}\binom{n}{i}\le2^{H(c)n}\text{ \, \, for }c\le\tfrac{1}{2}\, .
\end{equation}
It follows from \eqref{eqsparse} and \eqref{eqHtaylor} that for $\Delta$ sufficiently large 

\begin{eqnarray*}
\sum_{I \text{ sparse}}\lam^{|I|}&\le& \binom{n}{\le\tfrac{4\log\Delta}{\Delta}n}(1+\lam)^{\tfrac{4\log\Delta}{\Delta}n}\\
&\le&\exp\left\{ \log2\cdot H\left(\frac{4\log\Delta}{\Delta}\right)n+\frac{4\log\Delta}{\Delta}\log(1+\lam)n\right\}\\
&\le& \exp\left\{ 5\log2\cdot \frac{\log^2\Delta}{\Delta}n+\frac{4\log\Delta}{\Delta}\log(1+\lam)n\right\}\\
&\le& \exp\left\{\log(1+\lam)n/4\right\}\\
&=&(1+\lam)^{n/4}\, .
\end{eqnarray*}

It follows from \eqref{eqsparsesum} and the crude bound $Z^\ast_\cO(\lam)\ge(1+\lam)^{n/2}$ that
\[
\left| 1-\frac{\tilde Z_G(\lam)}{Z_G(\lam)}\right|<(1+\lam)^{-n/4}\, .
\]

\end{proof}

\subsection{Verifying the Koteck\'y-Preiss condition}

We will verify condition~\eqref{eqKPsimple} with the function $g(\gamma) = | \gamma| \frac{\Delta}{10 \log \Delta} \log (1+ \lam)$.   In particular, we will show 
\begin{align*}
\sum_{\gamma \ni v} w_{\gamma} e^{|\gamma| +g (\gamma)} \le  \frac{1}{\Delta^2} \,.
\end{align*}
We have 
\[
w_\gamma=\frac{\lam^{|\gamma|}}{(1+\lam)^{|\partial \gamma|}}\le {\lam^{|\gamma|}}{(1+\lam)^{-\tfrac{\Delta }{5\log \Delta}|\gam|}}\, ,
\]
and so proceeding as in Section~\ref{subsecVKP} we have
\begin{align*}
\sum_{\gamma \ni v} w_{\gamma} e^{|\gamma| +g (\gamma)}  &\le \sum_{k \ge 1} \exp \left[k \left( 2 \log \Delta +2  +  \frac{\Delta}{10 \log \Delta} \log (1+ \lam)  + \log \lam - \frac{\Delta}{5 \log \Delta} \log (1+ \lam)  \right)   \right] \\
&= \sum_{k \ge 1} \exp \left[k \left( 2 \log \Delta +2  -  \frac{\Delta}{10 \log \Delta} \log (1+ \lam)  + \log \lam  \right)   \right]
\end{align*}
which, for $\Delta$ large enough and $\lam \ge \frac{50 \log^2 \Delta}{\Delta}$, is at most $1/\Delta^2$. We may now finish as in Section~\ref{secHCfinish}, noting that this time  we can take $\rho=\Delta/(10 \log \Delta)$, $c_1=1$, and $c_2=0$ in the application of Theorem~\ref{thmPolymerCount} so that the run-time of the approximate counting algorithm is
$(n/\eps)^{1+O(\log^2 \Delta/\Delta)}$. This completes the proof of Theorem~\ref{thmHCrandombip}.

\section{Proper colorings}
\label{secColorProof}
In this section we prove Theorem~\ref{thmColor}. 
Let $G\in\cG^{\text{bip}}(n , \Delta)$. 
As in the proof of Theorem~\ref{thmHCrandombip},
the only property we require of $G$ is that it is a
$\left(\tfrac{4\log\Delta}{\Delta}, \tfrac{\Delta}{4\log\Delta}-\tfrac{1}{2}\right)$-expander
which holds with high probability. 
As before we denote the two partition classes of $G$ by $\cO, \cE$ and let $V=V(G)$.
Let $m=n/2$ (so that $|\cO|=|\cE|=m$) and fix an integer $q\geq3$.

Throughout this section all colorings will be proper vertex $q$-colorings.  Let $\cX=\cX_{G,q}$ be the set of all proper colorings $f:V(G)\to [q]$. For a set $S\subset V(G)$ and $f\in\cX$, we let $f(S):=\{f(v) : v\in S\}$ and we let $f|_S$ denote the restriction of $f$ to $S$, that is, the map $f|_S: S\to [q]$ where $f|_S(v)=f(v)$ for all $v\in S$.

 Our aim is to obtain an FPTAS for $Z_G(q)=|\cX|$ and an efficient sampling algorithm for $\mu_{G,q}$, the uniform distribution over $\cX$. 

We have the following important class of colorings of $G$, which play the role of ground states. 

\begin{defn}
Let $A,B$ be disjoint subsets of $[q]$ such that $A\cup B=[q]$. 
We call a coloring $f\in \cX$ an \emph{$(A,B)$-coloring}
if $f(\cO)\subset A$, $f(\cE)\subset B$. 
We call the pair $(A,B)$ a \emph{pattern}.
\end{defn}

This notion was inspired by the work of Peled and Spinka \cite{peled2018rigidity} where such patterns play a similar role. 

Let $\cP$ denote the set of all patterns. 
Given a subset $S\subset V$, a coloring $f\in \cX$ and a pattern $(A,B)$,
we say that $f$ \emph{agrees} with $(A,B)$ at $v\in V$ if 
$v\in\cO$ and $f(v)\in A$ or if $v\in\cE$ and $f(v)\in B$.
We say that $f$ disagrees with $(A,B)$ at $v$ otherwise.  
Let $\chi_{A,B}(S)$ be the set of colorings $f\in\cX$ such that 
$f$ \emph{disagrees} with $(A,B)$ at each $v\in S$ and agrees at each $v\in V\bs S$.

Suppose $(A,B)\in\cP$, $S\subset V$. 
We record the following simple bound for future use.  Recall that for $S\subset V$, $\partial S$ denotes the set of vertices in $V\bs S$ that are adjacent to a vertex of $S$.
\begin{lemma}\label{lemdbound}
For $(A,B)\in\cP$ and $S\subset V$ we have 
\[
|\chi_{A,B}(S)|\le
|A|^m|B|^m
\left(1-\frac{1}{|A|}\right)^{|\p S\cap \cO|}
\left(1-\frac{1}{|B|}\right)^{|\p S\cap \cE|}
\left(\frac{|A|}{|B|}\right)^{|S\cap\cE|-|S\cap\cO|}\, .
\]
\end{lemma}
\begin{proof}
For $f\in \chi_{A,B}(S)$,
there are at most $|A|^{|S\cap\cE|}|B|^{|S\cap\cO|}$ choices for $f|_S$
and given any such choice there are then at most
$(|A|-1)^{|\p S\cap \cO|}(|B|-1)^{|\p S\cap \cE|}$ choices for $f|_{\p S}$.
Finally, given any choice of $f|_{S^+}$
there are at most 
$|A|^{m-|S^+\cap \cO|}|B|^{m-|S^+\cap \cE|}$
choices for $f|_{V\bs S^+}$.
The result follows.
\end{proof}

Henceforth, let us call a set $S\subset V$ \emph{little} if
$|S| \le 4q\tfrac{\log\Delta}{\Delta}m$.

\begin{lemma}\label{lemalmost}
For every $f\in\cX$, there is a pattern $(A,B)$
and a little set $S\subset V$ such that
$f\in \chi_{A,B}(S)$.
\end{lemma}
\begin{proof}
Given $f\in\cX$ let 
\[
A'=\{i\in[q]: |f^{-1}(\{i\})\cap\cO|>4\tfrac{\log\Delta}{\Delta}m\}\, ,
\]
\[
B'=\{j\in[q]: |f^{-1}(\{j\})\cap\cE|>4\tfrac{\log\Delta}{\Delta}m\}\, .
\]
Note that since $G$ is a $\left(\tfrac{4\log\Delta}{\Delta}, \tfrac{\Delta}{4\log\Delta}-\tfrac{1}{2}\right)$-expander, we have that $E(X,Y)\neq\emptyset$ for any two subsets $X\subseteq \cO$, $Y\subseteq\cE$ such that $|X|, |Y|>4\tfrac{\log\Delta}{\Delta}m$.
It follows that $A'$ and $B'$ are disjoint.
Letting $(A,B)$ be any pattern such that $A'\subset A$ and $B'\subset B$,
we have $f\in \chi_{A,B}(S)$ for some little set $S$ by construction. 
\end{proof}

We now define a collection of polymer models and partition functions, a linear combination of which will serve as a good approximation to $Z_G(q)=|\cX|$. 
We define a polymer to be a \emph{little}, $G^3$-connected subset of $G$ (see Definition~\ref{defklink}).
We note that a subset $S\subset V$ is $G^3$-connected if and only if $G[S^+]$ is connected.
We say that two polymers $\gam_1, \gam_2$ are \emph{compatible} if $d_G(\gam_1, \gam_2)>3$
(i.e. $\gam_1\cup\gam_2$ is not $G^3$-connected).
We let $\cC=\cC(G)$ denote the set of all polymers of $G$ and let $\cG=\cG(G)$ denote the family of all sets of mutually compatible polymers from $\cC$. 
Let us now fix a pattern $(A,B)\in\cP$.
To each polymer $\gamma\in\cC$,
we assign a weight
\[
 w_{A,B}(\gam):=\frac{|\chi_{A,B}(\gam)|}
 {|A|^{m} |B|^{m}}\, .
\]
This defines a polymer model with partition function
\[
\Xi_{A,B}(G)=\sum_{\Gam\in \cG}\prod_{\gam\in\Gam} w_{A,B}(\gam)\, .
\]

\begin{lemma}\label{lemcolapprox}
There is an absolute constant $C$ such that if $\Delta\ge C q^2 \log^2 q$, then the sum
\[
\tilde Z_G(q):=
\sum_{(A,B)\in\cP}|A|^m|B|^m\cdot \Xi_{A,B}(G)
\]
is a $e^{-m/4q}$-relative approximation to $Z_G(q)$.
\end{lemma}
\begin{proof}
Let us fix $(A,B)\in\cP$ and let $S\subset V$.
Let $\hat \chi_{A,B}(S)$ be the set of colorings $f|_{S^+}$
where $f\in \chi_{A,B}(S)$.
Given any $g\in \hat \chi_{A,B}(S)$,
we may extend $g$ to an element of $\chi_{A,B}(S)$ 
by arbitrarily assigning vertices of $\cO\bs S^+$ with colors from $A$
and arbitrarily assigning vertices of $\cE\bs S^+$ with colors from $B$.
It follows that 
\begin{equation}\label{eqextend}
\frac{|\chi_{A,B}(S)|}{|A|^m|B|^m}=
\frac{|\hat \chi_{A,B}(S)|}{|A|^{|S^+\cap \cO|}|B|^{|S^+\cap \cE|}}\, .
\end{equation}
If the $G^3$-connected components of $S$ are
$\gamma_1,\ldots, \gamma_k$ we therefore have
\[
\frac{|\chi_{A,B}(S)|}{|A|^m|B|^m}=\prod_{i\in[k]}w_{A,B}(\gamma_i)\, .
\]
We call a set $S\subset V$ \emph{sparse}
if all of its $G^3$-connected components are little. 
We note that there is a one-one correspondence
between sparse subsets of $V$
and collections of mutually compatible polymers
(i.e. elements of $\cG$).
It follows that
\[
\sum_{S\text{ sparse}}|\chi_{A,B}(S)|=|A|^{m}|B|^{m}
\sum_{\Gam\in \cG}\prod_{\gam\in\Gam} w_{A,B}(\gamma)
=|A|^{m}|B|^{m}\cdot\Xi_{A,B}(G)\, 
\]

and so
\[
\tilde Z_G(q)=\sum_{(A,B)\in\cP}\sum_{S\text{ sparse}}|\chi_{A,B}(S)|\, .
\]
We first show that $\tilde Z_G(q)$ is a good approximation to the sum 
\[
\hat Z_G(q):=\sum_{(A,B)\in\cP}\sum_{S\text{ little}}|\chi_{A,B}(S)|\, .
\]
To this end we need the following claim.
\begin{claim}\label{claimsparse}
If $S$ is sparse
then $|S|\le12q\tfrac{\log\Delta}{\Delta}m$.
\end{claim}
\begin{proof}[Proof of Claim~\ref{claimsparse}]
Let $S$ be sparse and let 
$\Gam=\{\gam_1,\ldots, \gam_k\}$ be the $G^3$-connected components of $S$.
Suppose that 
$|S|>12q\tfrac{\log\Delta}{\Delta}m$.
Since each $\gamma_i$ is little by assumption,
we may partition $\Gam=\Gam_1\cup \Gam_2$ in such a way that
$\sum_{\gamma\in \Gam_j}|\gam|\ge\tfrac{4q \log \Delta}{\Delta}m$ for $j=1,2$.
Let $S_j=\bigcup_{\gam\in \Gam_j }\gam$ for $j=1,2$.
Suppose without loss of generality that 
$|S_1\cap\cO|\ge\tfrac{2q \log \Delta}{\Delta}m$.
Since $G$ is a 
$\left(\tfrac{4\log\Delta}{\Delta}, \tfrac{\Delta}{4\log\Delta}-\tfrac{1}{2}\right)$-expander, 
by considering a subset of $S_1\cap \cO$ of size $\tfrac{4\log\Delta}{\Delta}m$, we have 
\begin{align}\label{eqsubsetexp}
|\partial(S_1\cap\cO)|>\left(1-\frac{2\log\Delta}{\Delta}\right)m
\end{align} 
and so
$|S_1^+|>m$.
Similarly $|S_2^+|>m$ so that
$S_1^+\cap S_2^+\neq \emptyset$.
It follows that $\gam_i^+\cap\gam_j^+\neq\emptyset$
for some $\gam_i\in\Gam_1$, $\gam_j\in\Gam_2$, 
contradicting the fact that $\gam_i, \gam_j$ 
are distinct $G^3$-connected components of $S$.
\end{proof}

Suppose now that $S\subset V$ is sparse and not little, and
suppose without loss of generality that 
$|S\cap\cO|\ge 2q\tfrac{\log\Delta}{\Delta}m$.
As in \eqref{eqsubsetexp}
we have 
$|\partial(S\cap\cO)|\ge\left(1-\tfrac{2\log\Delta}{\Delta}\right)m$
and so by Claim~\ref{claimsparse}
 \[
 |\p S\cap\cE| >\left(1-\frac{14q\log\Delta}{\Delta}\right)m>\frac{m}{2}\, .
 \]
By Lemma~\ref{lemdbound} we then have 
 \begin{align}
 |\chi_{A,B}(S)|&
 \le |A|^m|B|^m\left(1-\frac{1}{|A|}\right)^{|\p S\cap \cO|}
 \left(1-\frac{1}{|B|}\right)^{|\p S\cap \cE|}
 \left(\frac{|A|}{|B|}\right)^{|S\cap\cE|-|S\cap\cO|}\\
 &\le |A|^m|B|^m
 \left(1-\frac{1}{|B|}\right)^{m/2}q^{12q\tfrac{\log\Delta}{\Delta}m}\, ,
 \end{align}
where for the final inequality we have used the crude bounds 
$|A|/|B|\in [1/q,q]$, $||S\cap\cE|-|S\cap\cO||\le |S|$
and Claim~\ref{claimsparse}. We thus have

\begin{align}
\tilde Z_G(q)-\hat Z_G(q) &=\sum_{(A,B)\in\cP}\sum_{\substack{
S\text{ sparse,}\\ \text{not little}}}|\chi_{A,B}(S)|\, ,\\
&\le\sum_{(A,B)\in\cP}\binom{2m}{12q\tfrac{\log\Delta}{\Delta}m}
 |A|^m|B|^m
  \left(1-\frac{1}{|B|}\right)^{m/2}q^{12q\tfrac{\log\Delta}{\Delta}m}\, ,\\
 &\le \sum_{(A,B)\in\cP} |A|^m|B|^m
 \exp\left\{2\log 2  \cdot H\left(6q\frac{\log\Delta}{\Delta}\right)m + 12q\log q\frac{\log\Delta}{\Delta}m- \frac{m}{2q } \right\}\, ,\\
 & \le e^{-m/(3q)}\sum_{(A,B)\in\cP} |A|^m|B|^m\, ,
\end{align}
where for the first inequality we used the fact that there are at most 
$\binom{2m}{12qm{\log\Delta}/{\Delta}}$ sparse subsets of $V$ and
for the final inequality we used the assumed lower bound on $\Delta$.
Using the crude bound $\tilde Z_G(q) \ge \sum_{(A,B)\in\cP} |A|^m|B|^m$
 we see that $\hat Z_G(q)$ is a $e^{-m/(3q)}$-relative approximation to $\tilde Z_G(q)$.

By Lemma~\ref{lemalmost} we have
\begin{align}\label{eqfunion}
\cX=
\bigcup_{(A,B)\in\cP}\bigcup_{S \text{ little} } \chi_{A,B}(S)\, .
\end{align}
Suppose now that $(A,B), (C,D)\in \cP$,
$S,T\subset V$ and
$\chi_{A,B}(S)\cap \chi_{C,D}(T)\neq \emptyset$.
Unless $(A,B,S)= (C,D,T)$ we must have $(A,B)\neq(C,D)$.
WLOG suppose that $|A|=\max\{|A|, |B|, |C|, |D|\}$.
 Then $|A\cap C|\le |A| -1$ and $|B\cap D|\le|B|$.
If $S, T$ are both little we then have
\begin{align}
|\chi_{A,B}(S)\cap \chi_{C,D}(T)|&\le 
(|A|-1)^m |B|^m
q^{8q\tfrac{\log\Delta}{\Delta}m}\, \\
&\le \left(\frac{q-1}{2} \right)^{2m}
\exp\left\{8q\log q\frac{\log\Delta}{\Delta}m\right\}\, .
\end{align}
The first inequality comes from the fact that there are at most $(|A|-1)^m$
ways to color the vertices of $\cO$ that agree with both $(A,B)$ and $(C,D)$ 
(and so are colored with elements of $A\cap C$), there are at most  $|B|^m$
ways to color the vertices of $\cE$ that agree with both $(A,B)$ and $(C,D)$,
and there are at most $q^{8q\tfrac{\log\Delta}{\Delta}m}$ ways to color the 
vertices of $S$ and $T$.
We may therefore bound $|\cX|$ by inclusion-exclusion as follows: 
\begin{align}
\hat Z_G(q) -2^{2q}
\binom{2m}{4q\tfrac{\log\Delta}{\Delta}m}^2
\left(\frac{q-1}{2} \right)^{2m}
\exp\left\{8q\log q\frac{\log\Delta}{\Delta}m\right\}
 \le |\cX| 
 \le \hat Z_G(q)\, ,
\end{align}
where we have used the bound $|\cP|\le 2^q$ and that there are at most
$\binom{2m}{4qm{\log\Delta}/{\Delta}}$
little subsets of $V$.
Using the bound $\hat Z_G(q)\ge \left\lceil \tfrac{q}{2} \right\rceil^m\left\lfloor \tfrac{q}{2} \right\rfloor^m $,
it follows that $\hat Z_G(q)$ is a $e^{-m/q}$-relative approximation to $|\cX|$.
The result follows.
\end{proof}

We may therefore focus on approximating each partition function $\Xi_{A,B}(G)$ individually.
Henceforth let us fix a pattern $(A,B)\in\cP$.

We will verify conditions (i)--(iii) of Theorem~\ref{thmPolymerCount} for the polymer model defined on $\cC$ with weight function $w_{A,B}$ in order to obtain an FPTAS for $\Xi_{A,B}(G)$. 

Verifying condition (i) is essentially immediate. Given $ \gamma\subset V(G)$, determining  whether $\gamma \in \cC$ amounts to checking whether $\gamma$ is $G^3$-connected or equivalently whether $G[\gamma^+]$ is connected. This can be done in $O(\Delta |\gamma|)$ time by a depth-first search algorithm.
 By~\eqref{eqextend}, $w_{A,B}(\gamma)$ can be calculated in
 $e^{O(\Delta \log q \cdot |\gamma|)}$ time by checking all possible colorings of $\gamma^+$.

We now turn our attention verifying conditions (ii) and (iii) for an appropriate choice of function $g$.

\subsection{Verifying the Koteck\'y-Preiss condition}\label{subsecVKPcol}
For brevity we denote $w_{A,B}$ simply by $w$.
We choose 
\[ g(\gamma)= \frac{\Delta}{10q^2\log \Delta}|\gam|   \,.\]
It remains to show the Koteck\'y-Preiss condition holds.  That is, 

\begin{align*}\label{eqKPsimple}
\sum_{\gamma':d(\gamma', \gamma)\leq3} 
w(\gamma')
 e^{| \gamma'| +g(\gamma')}\le | \gamma | 
\end{align*}

for all $\gamma\in \cC$.
If we could show that for each $v\in V(G)$
\begin{equation}\label{eqvertexcol}
\sum_{\gamma':\gamma'\ni v }w(\gam') \cdot e^{|\gamma'|+g(\gamma')}
\le \frac{1}{\Delta^3}\, ,
\end{equation}
then by summing this inequality over all $v$ at distance at most 3 from $\gam$ in $G$ (noting that there are at most $\Delta (\Delta-1)^2 +1 \le\Delta^3|\gamma|$ such vertices) we would be done.

The graph $G^3$ has maximum degree at most $\Delta^3$
and so by Lemma~\ref{lemConCount} the number of $G^3$-connected
sets of size $t$ containing vertex $v$ is at most $(e\Delta^3)^t$.
Suppose now that $S\subset V$ is little. 
Since $G$ is a $\left(\tfrac{4\log\Delta}{\Delta}, \tfrac{\Delta}{4\log\Delta}-\tfrac{1}{2}\right)$-expander, by considering a subset $T\subset S\cap \cO$ such that $|T|=\tfrac{|S\cap\cO|}{q}\le \tfrac{4\log\Delta}{\Delta}m$ we have 
\begin{align}
|\partial(S\cap \cO)|\ge \left(\frac{\Delta}{4\log\Delta}-\frac{1}{2}\right)\frac{|S\cap\cO|}{q}
\end{align}
and similarly
\begin{align}
|\partial(S\cap \cE)|\ge \left(\frac{\Delta}{4\log\Delta}-\frac{1}{2}\right)\frac{|S\cap\cE|}{q}\, .
\end{align}
By summing the above two inequalities we obtain
\begin{align}
|\p S|\ge |\partial(S\cap \cO)|+|\partial(S\cap \cE)|-|S|\ge \left(\frac{\Delta}{4\log\Delta}-q-\frac{1}{2}\right)\frac{|S|}{q}\, .
\end{align}
If $\gam\in\cC$, so that in particular $\gam$ is little,
by Lemma~\ref{lemdbound} we then have
\begin{align}
w(\gam)\le \left(1-\frac{1}{q}\right)^{|\p \gam|}q^{|\gam|}
\le \exp\left\{-\frac{\Delta}{5q^2\log \Delta}|\gam|\right\}\, ,
\end{align}
where we've assumed that $\Delta> Cq^2\log^2q$ for a large absolute constant $C$ as in Lemma~\ref{lemcolapprox}. Putting everything together we have
\begin{align}\label{eqKPcol}
\sum_{\gamma':\gamma'\ni v }w(\gam') \cdot e^{|\gamma'|+g(\gamma')}\le
 \sum_{t=1}^{\infty}(e^2\Delta^3)^t
  \exp\left\{-\frac{\Delta}{10q^2\log \Delta}t\right\}\leq\frac{1}{\Delta^3}\, .
\end{align}

\subsection{Proof of Theorem~\ref{thmColor}}
\label{secColfinish}

We consider two cases separately.  If $\eps < e^{-n/(8q)}$, then we proceed by brute force, checking all possible elements of $[q]^{V(G)}$ to see if they are a proper $q$-coloring. In this way we can calculate $Z_G(q)$ exactly in time $O(n\Delta q^n)$ and therefore count and sample in time polynomial in $1/\eps$. 

Now we assume $\eps > e^{-n/(8q)}$ and take 
\begin{align}
\label{eqDelbound}
\Delta> C q^2 \log^2 q \,,
\end{align}
where $C$ is chosen large enough so that Lemmas~\ref{lem:bassexp}, \ref{lemcolapprox}  and inequality \eqref{eqKPcol} all hold.

For each $(A,B)\in \cP$, using the FPTAS for $ \Xi_{A,B}(G)$ given by Theorem~\ref{thmPolymerCount}, we may find $Z_{A,B}^{\text{alg}}$, an $\eps/3$-relative approximation to $ |A|^m|B|^m \Xi_{A,B}(G)$ in time polynomial in $n$ and $1/\eps$.
It follows that $Z^{\text{alg}}:=\sum_{(A,B)\in\cP}Z_{A,B}^{\text{alg}}$ is an $\eps/3$-relative approximation to $\tilde Z_G(q)$ (as defined in Lemma~\ref{lemcolapprox}). By Lemma~\ref{lemcolapprox}, $\tilde Z_G(q)$ is an $\eps/3$-relative approximation to $Z_G(q)$ and so $Z^{\text{alg}}$ is an $\eps$-relative approximation to $Z_G(q)$ as required.

We now turn our attention to the sampling algorithm.
Let $\mu$ denote the uniform distribution on $\cX$.
Consider the distribution $\hat \mu$ on  $\cX$ defined as follows. First choose $(A,B)\in\cP$ with probability proportional to 
$|A|^m|B|^m\Xi_{A,B}(G)$. Then, sample $\Gam\in\cG$  from the measure
\begin{align*}
\nu_{A,B}(\Gamma) &= \frac{ \prod_{\gamma \in \Gamma} w_{A,B}(\gamma) }{ \Xi_{A,B}(G) }\, .   
\end{align*}
Let $S=\bigcup_{\gamma\in\Gam}\gam$.
Uniformly choose a coloring of $S^+$
which disagrees with $(A,B)$ on $S$ and agrees on $\p S$.  Then extend this coloring to a coloring of $G$ by
coloring vertices of $\cO\bs S^+$ with elements of $A$ 
and vertices of $\cE\bs S^+$ with elements of $B$ uniformly at random.
The distribution of the resulting coloring is $\hat \mu$.
By Lemma~\ref{lemcolapprox} we have
\begin{align*}
\| \hat \mu - \mu \|_{TV} \le \eps/2 \, .
\end{align*}
Then to obtain an $\eps/2$-approximate sample from $\hat \mu$ efficiently, we proceed as in Section~\ref{secHCfinish}- first we approximate $ \Xi_{A,B}(G)$ for $(A,B) \in \cP$ in order to approximately sample $(A,B)\in\cP$ and then we apply Theorem~\ref{thmPolymerSample} to obtain an approximate sample from $\nu_{A,B}$.  As in Section~\ref{secPottsproof},  the polymer sampling algorithm of Theorem~\ref{thmPolymerSample} guarantees that $|S^+| = O( \log(n/\eps))$ and so a uniformly random coloring of $S^+$ which disagrees with $(A,B)$ on $S$ and agrees on $\p S$ can be obtained by brute force in time polynomial in $n$ and $1/\eps$.

\qed

\section{Concluding Remarks}
\label{secconclude}
The algorithms presented here are the first provably efficient counting and sampling algorithms for \#BIS-hard problems for the class of expander graphs. However, they are presumably not optimal in terms of either their running time or the range of parameters for which they are provably efficient.

One natural choice for more efficient algorithms would be those based on Markov chains.  A  candidate algorithm for the Potts model is the Swendsen-Wang dynamics~\cite{swendsen1987nonuniversal}.  It is natural to conjecture that the Swendsen-Wang dynamics are rapidly mixing on expander graphs at sufficiently low temperatures, and this would give a more efficient sampling algorithm than the one presented here.  Similarly, for the hard-core model on a bipartite expander graph with symmetry between the sides of the bipartition, one could follow the suggestion of~\cite{helmuth2018contours} and start the Glauber dynamics in either the all even or all odd occupied state with equal probability.  Proving that such sampling algorithms are indeed efficient is left as an open problem.  Recently Chen, Galanis, Goldberg, Perkins, Stewart, and Vigoda~\cite{PolymerMarkov} have made some progress in this direction by showing that a restricted version of Glauber dynamics for the Potts and hard-core models mixes in polynomial time under similar conditions to those of Theorems~\ref{thmhard-core} and~\ref{thm:pottspartition}.

For random graphs, we can ask if efficient algorithms for the problems addressed in this paper exist for \textit{all} ranges of parameters.  

\begin{question}
Is there a polynomial-time sampling algorithm for the Potts model on random regular graphs at all inverse temperatures $\beta$?  Are there efficient sampling algorithms for the hard-core model and proper colorings on random regular bipartite graphs at all fugacities $\lam$ and for all $q$?
\end{question}

\end{document}